\documentclass[10pt]{article}

\usepackage[margin=1in]{geometry}  
\usepackage{graphicx}              
\usepackage{amsmath}               
\usepackage{amsfonts}              
\usepackage{amsthm}
\usepackage{tikz}               
\usepackage{algorithm}
 \usepackage{algpseudocode,multirow,booktabs,caption,siunitx}
\usetikzlibrary{arrows,positioning}
\usetikzlibrary{chains,fit,shapes}
\usetikzlibrary{calc}
\usepackage{pgf}
\usepackage{xxcolor}
\usepackage{subfigure}
\usetikzlibrary{arrows,shadows,petri}

\newtheorem{thm}{Theorem}[section]
\newtheorem{lem}[thm]{Lemma}

\newtheorem{cor}[thm]{Corollary}

\newtheorem{observation}[thm]{Observation}

\DeclareMathAlphabet{\mathbmit}{OML}{cmm}{b}{it}
\renewcommand{\vec}[1]{\mathbmit{#1}}
\let\matr\vec

\begin{document}

\date{}
\nocite{*}

\title{The generalized vertex cover problem and some variations}

\author{
Pooja Pandey \thanks{ Corresponding author. Email: poojap@sfu.ca } and Abraham P. Punnen\\
Department of Mathematics, Simon Fraser University\\
 250 - 13450 – 102nd Avenue, Surrey, BC, V3T 0A3, Canada}

\maketitle

\begin{abstract}
In this paper  we study  the generalized vertex cover problem (GVC), which is a generalization of various well studied combinatorial optimization problems.  GVC is shown to be equivalent to the unconstrained binary quadratic programming problem  and also equivalent  to some other variations of the general GVC. Some solvable cases are identified and approximation algorithms are suggested for special cases. We also study GVC on bipartite graphs and identify some polynomially solvable cases. We show that GVC on bipartite graphs is equivalent to the bipartite unconstrained 0-1 quadratic programming problem. Integer programming formulations of GVC and related problems are presented and establish half-integrality property on some variables for the corresponding linear programming relaxations. We  also discuss special cases of GVC where all feasible solutions are independent sets or vertex covers. These problems are observed to be equivalent to the maximum weight independent set problem or minimum weight vertex cover problem along with some algorithmic results.
\end{abstract}



\section{Introduction}

Let $G=(V,E)$ be a  graph with $V=\{1,2,\hdots,n\} \text{ and } |E|=m$. For each edge $(i,j) \in E$ three real valued weights $q_{ij}^0, q_{ij}^1, \text{ and }  q_{ij}^2 $ are associated. Also, for each vertex $i \in V$ a weight $c_i$ is prescribed. For any subset $ U \subseteq V$, let $E_0(U) = \{(i,j) \in E: i,j \not\in U \}$, $E_1(U) = \{(i,j) \in E: \text{ either } i \in U \text{ or } j \in U \text{ but not both} \} $,  $E_2(U) = \{(i,j)\in E: i,j \in U \}$, and
$$f(U) = \sum_{i \in U}  c_i + \sum_{(i,j) \in E_0(U)}q_{ij}^0  + \sum_{ (i,j) \in E_1(U)}q_{ij}^1 + \sum_{(i,j) \in E_2(U)}  q_{ij}^2. $$
Then the {\it generalized vertex cover problem} (GVC) is to find a set $ U \subseteq V$ such that $f(U)$ is minimized. Note that $q_{ij}^0$ can be viewed as the `cost of' not covering edge $(i,j)$, $q_{ij}^1$ can be viewed as the `cost of' covering  edge $(i,j)$  by selecting exactly one of its end points, and $q_{ij}^2$ can be viewed as the `cost of' over-covering edge $(i,j)$. If the solution is an empty set $\phi$, then the objective function is defined as $f(\phi) = \sum_{(i,j) \in E}q_{ij}^0$. \\

GVC was introduced by Hassin and Levin \cite{2003Rafael} and it is a meaningful generalization of the well known  minimum weight vertex cover  problem (MWVCP) \cite{1982Dorit,2001Varizani} and the maximum weight independent set  problem (MWISP) \cite{1994Baker,2001Varizani}. Note that if $c_i\geq 0$ for all $i\in V,$ $q^0_{ij}=M$, a large number and  $q^1_{ij}=q^2_{ij}=0$ for all $(i,j)\in E$ then, GVC reduces to MWVCP. Likewise, when $c_i\leq 0$, $q^2_{ij}= M$, $q^1_{ij}=q^0_{ij}=0$ for all $(i,j)\in E$, then GVC is equivalent to MWISP. (It may be noted that MWISP is normally presented as a maximization problem which is equivalent to minimization form indicated above). Hassin and Levin \cite{2003Rafael} although introduced GVC,   they focused primarily on a  special case of it where $c_i$ is assumed to be non-negative for all $i\in V$ and $q^0_{ij} \geq q^1_{ij}\geq q^2_{ij} \geq 0$ for all $(i,j)\in E$. We denote this special case of GVC by GVC-HL and it may be noted that MWISP is not a special case of  GVC-HL. In \cite{2003Rafael} two 2-approximation algorithms for GVC-HL are proposed , one  based on linear programming, and the other based on the local-ratio technique \cite{1985Rauven}. When $q^1_{ij}= \alpha,~ ( 0 \leq \alpha \leq 1),~~  q^0_{ij}= 1, q^2_{ij}=0$ for all $(i,j)\in E$ and $c_i = \beta, \forall i \in V$, GVC is called \textit{uniform cost generalized vertex cover problem} (UGVC). In \cite{2003Rafael}  the complexity of UGVC has been studied  for all possible values of $\alpha $ and $ \beta$. They showed that GVC is polynomial time solvable in the following cases: 1)  $ \alpha \geq \frac{1}{2}$, 2) $ \alpha < \frac{1}{2} \text{ and } \beta \leq 3 \alpha$, 3)  $ \alpha < \frac{1}{2}$  and there exists an integer $ d \geq 3$ such that $d(1-\alpha) \leq \beta \leq (d+1)\alpha$. For the general case, UGVC is NP-hard \cite{2003Rafael}. Milanovic \cite{2010Marija} proposed a genetic algorithm to solve GVC-HL and reported  experimental results comparing their algorithm with CPLEX  and the 2-approximation  algorithm given in \cite{2003Rafael}.  Kochenberger et. al. \cite{2015Gary} compared an integer linear programming formulation  and an integer quadratic programming formulations using  the CPLEX solver.  \\

Another special case of GVC  where $q^1_{ij}=q^2_{ij}=0$ for all $(i,j)\in E$ was considered by Houchbaum \cite{2002Dorit} and Bar-Yehuda et al. \cite{2010Reuven}.  This problem is also known as the generalized vertex cover problem  in the literature and for definiteness we denote this problem by GVC1.  In fact, Houchbaum \cite{2002Dorit} and Bar-Yehuda et al. ~\cite{2010Reuven} studied primarily  a special case of GVC1 where $c_i\geq 0$ for all $i\in V$, $q^0_{ij} \geq 0$ for all $(i,j)\in E$.  We refer to this version of GVC1 as GVC1-HB. Houchbaum  \cite{2002Dorit} provided an integer linear programming formulation of GVC1-HB and showed that the corresponding linear programming relaxation admits half integrality property.  Bar-Yehuda et al. \cite{2010Reuven} provided an extension of a well known theorem by Nemhauser and Trotter \cite{1975Nemhauser} for the vertex cover problem  (independent set problem) to GVC1-HB, and presented a $(2-2/d)$-approximation algorithm  on graphs with maximum degree of a node  is bounded  above by $d$.  They also presented  a  polynomial time approximation scheme (PTAS) for GVC1-HB on  planner graphs, and a $(2-\log \log n/2 \log n)$- approximation algorithm for a general graph. In the same paper they showed  that GVC1-HB is NP-hard on complete graphs but solvable in polynomial time on   bipartite graphs. Note that MWVCP is trivial on complete graphs.\\

The generalized independent set problem introduced by  Hochbaum \cite{2002Dorit} is yet another special case of GVC. Here $q^1_{ij} \text{ and } q^0_{ij}  $  are assumed to be zero  for all $(i,j)\in E$, and we  denote this problem by GVC2.   Hochbaum and Pathria \cite{1997Dorit} studied a special case of GVC2 where  $c_i\leq 0$ for all $i\in V$, $q^2_{ij} \geq 0$ for all $(i,j)\in E$. This version of GVC2 is denoted by   GVC2-HP and the model have applications in Forest Harvesting. GVC2-HP in \cite{1997Dorit}  was presented as  a maximization problem with $c_i \geq 0, \forall i \in V$ and $q_{ij}^2 \leq 0, \forall (i,j) \in E$. Clearly,  this is equivalent to our definition of GVC2-HP. Kochenberger et al. \cite{2007Gary}  gave a nonlinear formulation for GVC2-HP and compare the computational effectiveness of this non-linear  integer programming formulation with an available linear integer programming formulation. \\

Another problem closely related to GVC is the unconstrained binary quadratic programming problem (UBQP)  studied by various authors \cite{2001Allemand, 2005Ferrez, 2010Kary,2014Gary, 1989padberg}.  Let $\matr{Q}=(q_{ij})$ be an $n \times n$ symmetric matrix and $\vec{a} = (a_1,\hdots, a_n)$ be an  $n$-vector. Then UBQP is to find  an $\vec{x}=(x_1,x_2,\hdots,x_n)  \in \{0,1\}^n$ such that
\begin{align*}
 &\sum_{i =1}^{n} a_ix_i + \sum_{i =1}^{n} \sum_{j =1}^{n} q_{ij} x_ix_j
\end{align*}

\noindent is minimized (or  maximized). Without loss of generality $q_{ii}$ is chosen as zero for $i=1,2,\hdots,n$ and we assume UBQP is presented as a minimization problem. Also the assumption that $\matr{Q}$ is symmetric is not a restriction, since $\matr{Q}$ can be replaced by $\dfrac{\matr{Q} + \matr{Q}^T}{2}$ to obtain  an equivalent problem.\\

Given the matrix $\matr{Q}$, let $E(\matr{Q})=\{(i,j) : q_{ij} \not= 0, i,j=1,\hdots,n \}$. The subgraph $G(\matr{Q})=(V,E(\matr{Q}))$ of the complete graph  $K_n$ with  the vertex set $V=\{1,2,\hdots,n\}$  is called the support graph of $\matr{Q}$. Then UBQP can be reformulated in terms of $G(\matr{Q})$ i.e; UBQP is equivalent to finding $U \subseteq V$ such that

\begin{align*}
&\sum_{i \in U} a_i + \sum_{ (i,j) \in E_2(\matr{Q},U)}  2q_{ij}
\end{align*}

\noindent is minimized (or  maximized), where $E_2(\matr{Q},U)= \{(i,j)\in E(\matr{Q}): i,j \in U \}$. \\

The bipartite unconstrained 0-1 quadratic program problem (BQP01)  \cite{2014Duarte, 2015Abraham,  2015Punnen1} is closely related to GVC on bipartite graphs.  Let $\matr{Q} = (q_{ij})$ be an  $m \times n$ matrix, $\vec{a} = (a_1, a_2, \hdots , a_m)$, $\vec{b} = (b_1, b_2,\hdots , b_n)$. Then BQP01 is defined as

\begin{align*}
 &\text{Minimize }  ~~  f(\vec{x},\vec{y}) =  \vec{x}^T \matr{Q} \vec{y} +\vec{a}\vec{x}+\vec{b}\vec{y}= \sum_{i=1}^{m} a_i x_i +  \sum_{j=1}^{n} b_j  y_j + \sum_{i=1}^{m} \sum_{j=1}^{n}q_{ij}x_iy_j \\
 &\text{ subject to: }  ~~\vec{x} \in \{0,1\}^m, \vec{y} \in \{0,1\}^n.
\end{align*}

Consider an instance of BQP01 with cost matrix $\matr{Q}$ of dimension $ m \times n$. The \textit{support  bipartite graph} of the matrix $\matr{Q}$ is the   bipartite graph $G[\matr{Q}]=(V_1,V_2,E[\matr{Q}])$, where $V_1=\{1,2,\hdots,m\}$, $V_2=\{1,2,\hdots,n\}$ and $E[\matr{Q}])=\{(i,j): i \in V_1, j \in V_2, q_{ij} \not= 0 \}$. The BQP01 can be  formulated as a graph theoretic optimization problem on $G[\matr{Q}])=(V_1,V_2,E[\matr{Q}]))$ as

\begin{align*}
\text{Minimized } ~ & \phi(U_1,U_2) = \sum_{i \in U_1} \sum_{j \in U_2} q_{ij} + \sum_{i \in U_1}   a_i + \sum_{j \in U_2}  b_j \\
\text{ subject to: }~ & U_1 \subseteq V_1, U_2 \subseteq V_2.
\end{align*}

It may be noted that the definition of support bipartite graph is different from that of a support graph. The support graph  when the  underlying graph is bipartite is different from a support bipartite graph.\\

In this paper  we study the general problems  GVC, GVC1, and GVC2.   We show that all these problems are equivalent to each other and also equivalent to UBQP. When the underlying graph is bipartite, these are equivalent to BQP01 as well.  However, it may be noted that although  these equivalences are verified in terms of optimality, domination ratio \cite{1997Fred}, differential approximation ratio \cite{1996Demange}, the  characteristics of these problems  in terms of approximation ratio \cite{2001Varizani} and various special cases are  different and hence  it is interesting to explore various special cases of these problems as well. We present several  complexity results related to GVC, GVC1, and GVC2 along with some polynomial solvable  special cases. Integer programming formulations of GVC, GVC1, and GVC2 are presented and establish half-integrality property on some variables for the corresponding linear programming relaxations. We  also present an approximation algorithm for  GVC with  the approximation ratio   $\max\{2,\alpha, \alpha \beta\}$  when  $ q^2_{ij} \leq \alpha q^1_{ij}$, $ q^1_{ij} \leq \beta q^0_{ij}$, and $\alpha, \beta \geq 1$ generalizing a result  given in  \cite{2003Rafael}. The approximation ratio for the  algorithm becomes $\max\{2,\gamma\}$  when for given $K \geq 0 \text{ and }  \gamma > 1$, and  for each edge $(i,j)$, all three weights $q_{ij}^0, q_{ij}^1, \text{ and }  q_{ij}^2 $ are in  $ [K,\gamma K]$.  When $q^{0}_{ij} = \infty$, the optimal solutions of GVC are vertex covers and this special case leads to the vertex cover problem with node and edge weights (VCPNEW). The problem VCPNEW and two of its special cases are shown to be equivalent to the minimum weight vertex cover problem (MWVCP) and shown that VCPNEW can be solved by an $\epsilon^{'}$ algorithm whenever  MWVCP can be solved by an $\epsilon$-algorithm for an appropriate  $\epsilon^{'} \leq \epsilon$ if $\sum_{(i,j) \in E} (2q_{ij}^1 - q_{ij}^2) \geq 0$ and $c_i + \sum_{(i,j) \in E} (q_{ij}^2 - q_{ij}^1) \geq 0$ for all $i\in V$.   When $q^{2}_{ij} = \infty$,  optimal solutions of GVC are independent sets and this special case leads to the independent set problem  with node and edge weights (ISPNEW). The problem ISPNEW and two of its special cases are shown to be equivalent to the maximum weight independent set problem (MWISP).\\

The paper is organized as follows: In Section 2, we study the complexity of GVC. Different NP-hard special cases as well as polynomially solvable special cases are identified. Further we show that GVC, GVC1, GVC2, and UBQP are pairwise equivalent in the sense that an algorithm to compute an optimal solution to one problem can be used to compute an optimal solution to the other. We also show that GVC, GVC1, and GVC2 on bipartite graphs are equivalent to BQP01.  Section 3 contains integer programming formulation of GVC and related problems, establish half-integrality property on some variables, and also discuss  the approximation algorithm for GVC. Section 4 discusses  VCPNEW, ISPNEW and some of their special cases. We also establish that  VCPNEW is equivalent to MWVCP, and ISPNEW is equivalent to MWISP.

\section{Complexity and Solvable cases}

We first show that the problems GVC, GVC1, GVC2, and UBQP are equivalent from an optimality point of view,  i.e.,   any one of these problems can be formulated as another in the sense  that from an optimal solution to one, an optimal solution  to another can be recovered.

\begin{thm} GVC, GVC1, and GVC2 are equivalent. \label{gen2}
\end{thm}
\begin{proof}  We first show that GVC can be formulated as GVC1. Given an instance of GVC, define

\begin{align}
    c_i^{'} = ~ & c_i + \sum_{(i,j) \in E}(q_{ij}^2-q_{ij}^1),~ \forall ~ i \in V  \text{ and }\\
     q_{ij}^{0^{'}}  =~ & q_{ij}^2-2q_{ij}^1+ q_{ij}^0,~ (i,j) \in E.
\end{align}

\noindent  Now, consider the instance of GVC1 where $c_i$ is replaced by  $c_i^{'}$ and  $q_{ij}^{0}$ is replaced by $q_{ij}^{0^{'}}$ . Then, for any $U \subseteq V$, the objective function of this GVC1  can be written as

 \begin{equation*}
  g(U)= \sum_{ i\in U} c_i^{'}  + \sum_{(i,j) \in E_0(U)}q_{ij}^{0^{'}} =   \sum_{ i\in U} \left(c_i+ \sum_{(i,j) \in E}(q_{ij}^2-q_{ij}^1) \right) + \sum_{(i,j) \in E_0(U)} ( q_{ij}^2-2q_{ij}^1+ q_{ij}^0).
 \end{equation*}

The objective function $f(U)$ of GVC is given by

\begin{align}
\label{GVC2UBQP} f(U) =& \sum_{i \in U}  c_i + \sum_{ (i,j) \in E_0(U)}q_{ij}^0 + \sum_{ (i,j) \in E_1(U)}q_{ij}^1 + \sum_{(i,j) \in E_2(U)}  q_{ij}^2.
\end{align}

It can be verified that
\begin{eqnarray}
\label{GVC2UBQP1}  \sum_{ (i,j) \in E_2(U)} q^{2}_{ij} &=&   \sum_{ (i,j) \in E, i\in U} q^{2}_{ij} + \sum_{ (i,j) \in E, j\in U} q^{2}_{ij} + \sum_{ (i,j) \in E_0(U)} q^{2}_{ij} - \sum_{ (i,j) \in E} q^{2}_{ij}  \text{ and }\\
 \label{GVC2UBQP2}  \sum_{ (i,j) \in E_1(U) } q^{1}_{ij} & =&  \sum_{ (i,j) \in E} 2q^{1}_{ij} -\sum_{ (i,j) \in E_0(U)} 2q^{1}_{ij} - \sum_{ (i,j) \in E, i\in U} q^{1}_{ij} - \sum_{ (i,j) \in E, j\in U} q^{1}_{ij}.
\end{eqnarray}

From (\ref{GVC2UBQP}), (\ref{GVC2UBQP1}) and (\ref{GVC2UBQP2}), we have

\begin{align*}
 f(U) =& \sum_{i \in U}  c_i + \sum_{ (i,j) \in E_0(U)}q_{ij}^0 +  \sum_{ (i,j) \in E} 2q^{1}_{ij} -\sum_{ (i,j) \in E_0(U)} 2q^{1}_{ij} - \sum_{ (i,j) \in E, i\in U} q^{1}_{ij} - \sum_{ (i,j) \in E, j\in U} q^{1}_{ij}\\
 &+ \sum_{ (i,j) \in E, i\in U} q^{2}_{ij} + \sum_{ (i,j) \in E, j\in U} q^{2}_{ij} + \sum_{ (i,j) \in E_0(U)} q^{2}_{ij} - \sum_{ (i,j) \in E} q^{2}_{ij},\\
\nonumber   = & \sum_{i \in U} \left(c_i+ \sum_{(i,j) \in E}(q_{ij}^2-q_{ij}^1) \right)  + \sum_{ (i,j) \in E_0(U)} ( q_{ij}^2-2q_{ij}^1+ q_{ij}^0)  +  \sum_{(i,j) \in E} (2q_{ij}^1- q_{ij}^2),  \\
\nonumber    = & g(U) + \sum_{(i,j) \in E} (2q_{ij}^1- q_{ij}^2).
\end{align*}

\noindent Thus, for any feasible solution $U \subseteq V$ of GVC and the instance of GVC1 constructed above, $f(U) - g(U) $ is a constant. Therefore, an optimal solution to this GVC1   will also be  an optimal solution to GVC.\\

Since GVC1 is a special case of GVC, any GVC1 can be formulated as GVC, establishing equivalence between GVC and GVC1.

To establish the equivalence between GVC and GVC2, define
\begin{align}
    c_i^{''} = ~& c_i+ \sum_{(i,j) \in E}(q_{ij}^1-q_{ij}^0),~ \forall ~ i \in V  \text{ and } \\
     q_{ij}^{2^{''}}  = ~ & q_{ij}^2-2q_{ij}^1+ q_{ij}^0,~ (i,j)  \in E.
\end{align}

\noindent Consider the instance of  GVC2 where $c_i$ is replaced by  $c_i^{''}$ and  $q_{ij}^{2}$ is replaced by $q_{ij}^{2^{''}}$. As in the previous case, we can show that an optimal solution to this GVC2 provides an optimal solution to GVC and the converse follows the fact that GVC2 is a special case of GVC.\\
\end{proof}

\begin{cor}GVC, GVC1, and GVC2 are NP-hard on complete graphs.
\end{cor}
\begin{proof} This follows from Theorem \ref{gen2} and the fact that GVC2-HB is NP-hard on complete graphs \cite{2010Reuven}.
\end{proof}

It may be noted that the special cases MWISP, MWVCP of GVC are trivial on a complete graph.

\begin{thm} GVC, GVC1, and GVC2 are equivalent to UBQP. \label{gen3}
\end{thm}

\begin{proof} We first show that  GVC2 can be formulated as UBQP.  For any feasible solution $U \subseteq \{1,2,\hdots,n\}$, let $ h(U)$ and $\phi(U)$  be the objective function values of  GVC2, and UBQP respectively.\\

Define the matrix $\matr{Q} = (q_{ij})_{n \times n}$ as

 \begin{equation}
q_{ij} =
\begin{cases}
\dfrac{q_{ij}^2}{2} ,  & \forall (i,j) \in E,\\
0  &\text{ if }  (i,j) \not\in E.
\end{cases} \label{QGVC2}
\end{equation}
and
\begin{equation}
a_i = c_i, ~ \forall i \in V. \label{aGVC2}
\end{equation}

\noindent  Consider  the UBQP with $\matr{Q}$ and $\vec{a}$  as defined above. For any feasible solution $ U \subseteq  \{1,2\hdots,n\}$ of GVC2 or UBQP, $h(U)=\phi(U)$. Thus an optimal solution to UBQP constructed above  is also an optimal solution to GVC2. \\

To show  that UBQP can be formulated as GVC2, let $\matr{Q}$ be the cost matrix of an instance  of UBQP and $G(\matr{Q})$ be the support graph of $\matr{Q}$. Choose $c_i = a_i, ~i=1,2, \hdots,n$, $q^2_{ij} = 2q_{ij},~ i,j=1,2,\hdots,n$ for  edge $(i,j)$ of the support graph $G(\matr{Q})$. Then, an optimal solution to the  resulting GVC2 on $G(\matr{Q})$ solves  UBQP. Thus GVC2 and UBQP are equivalent. The equivalence of GVC, GVC1, and UBQP now follow from Theorem \ref{gen2}.
\end{proof}

Although the equivalence of GVC, GVC1,  GVC2 and  UBQP are established in Theorem \ref{gen3}, the direct reformulation as a UBQP was given only for the GVC2. It is however interesting to present the precise structure of $\matr{Q}$ when GVC and GVC1 are reformulated as a UBQP. This is particularly useful in identifying polynomially solvable special cases directly from solvable UBQP instances. For any feasible solution $U \subseteq \{1,2,\hdots,n\}$, let $f(U), g(U), $ and $\phi(U)$  be the objective function values of GVC, GVC1, and UBQP respectively.\\

Define $\matr{Q} = (q_{ij})_{n \times n}$ as

 \begin{equation}
q_{ij} =
\begin{cases}
\dfrac{q_{ij}^2 - 2q_{ij}^1 + q_{ij}^0}{2},  & \forall (i,j) \in E,\\
0  &\text{ if }  (i,j) \not\in E. \label{QGVC}
\end{cases}
\end{equation}

and
\begin{equation}
a_i = c_i + \sum_{(i,j) \in E} \left( q_{ij}^1 - q_{ij}^0 \right) ~ \forall i\in V. \label{aGVC}
\end{equation}
Consider the UBQP with $\matr{Q}$ and $\vec{a}$ are as defined above. For any  $ U \subseteq  \{1,2\hdots,n\}$,\\

\begin{equation}
\label{GVCUBQP}   f(U) =  \sum_{i \in U} c_i + \sum_{ (i,j) \in E_0(U)} q^{0}_{ij} + \sum_{ (i,j) \in E_1(U) } q^{1}_{ij} + \sum_{ (i,j) \in E_2(U)} q^{2}_{ij}. \\
\end{equation}

It can be verified that
\begin{eqnarray}
\label{GVCUBQP1}  \sum_{ (i,j) \in E_0(U)} q^{0}_{ij} &=& \sum_{ (i,j) \in E} q^{0}_{ij} - \sum_{ (i,j) \in E, i\in U} q^{0}_{ij} - \sum_{ (i,j) \in E, j\in U} q^{0}_{ij} + \sum_{ (i,j) \in E_2(U)} q^{0}_{ij}  \text{ and }\\
 \label{GVCUBQP2}  \sum_{ (i,j) \in E_1(U) } q^{1}_{ij} & =&  \sum_{ (i,j) \in E, i\in U} q^{1}_{ij} + \sum_{ (i,j) \in E, j\in U} q^{1}_{ij} -  \sum_{ (i,j) \in E_2(U)} q^{1}_{ij}.
\end{eqnarray}

From (\ref{GVCUBQP}), (\ref{GVCUBQP1}) and (\ref{GVCUBQP2}), we have
 \begin{align}
\nonumber f(U) =&  \sum_{i \in U} c_i + \sum_{ (i,j) \in E} q^{0}_{ij} - \sum_{ (i,j) \in E, i\in U} q^{0}_{ij} - \sum_{ (i,j) \in E, j\in U} q^{0}_{ij} + \sum_{ (i,j) \in E_2(U)} q^{0}_{ij},  \\
\nonumber  & +  \sum_{ (i,j) \in E, i\in U} q^{1}_{ij} + \sum_{ (i,j) \in E, j\in U} q^{1}_{ij} -  \sum_{ (i,j) \in E_2(U)} q^{1}_{ij} + \sum_{ (i,j) \in E_2(U)} q^{2}_{ij}, \\
 \nonumber = & \sum_{i \in U} (c_i + \sum_{(i,j) \in E} \left( q_{ij}^1 - q_{ij}^0 \right)) + \sum_{ (i,j) \in E_2(U)} (q_{ij}^2 - 2q_{ij}^1 + q_{ij}^0) + \sum_{(i,j) \in E} q_{ij}^0, \\
 \nonumber = & \sum_{i \in U} a_i + \sum_{ (i,j) \in E_2(U)} 2q_{ij} +  \sum_{(i,j) \in E} q_{ij}^0, \\
 \label{GVCUBQP1.1}  =&  \phi(U) + \sum_{(i,j) \in E} q_{ij}^0.
\end{align}

 Thus an optimal solution to the UBQP constructed above is also an optimal solution to GVC.\\

Now, define the matrix $\matr{Q} = (q_{ij})_{n \times n}$ as

\begin{equation}
q_{ij} =
\begin{cases}
\dfrac{q_{ij}^0}{2} ,  & \forall (i,j) \in E,\\
0  &\text{ if }  (i,j) \not\in E.
\end{cases} \label{QGVC1}
\end{equation}

and
\begin{equation}
a_i = c_i- \sum_{(i,j) \in E}q_{ij}^0, ~ \forall i \in V. \label{aGVC1}
\end{equation}
 Consider an instance  of UBQP with $\matr{Q}$ and $\vec{a}$ as defined  above.\\

 Then
\begin{eqnarray*}
  g(U) &=& \sum_{i \in U} c_i + \sum_{ (i,j) \in E_0(U)} q^{0}_{ij}, \\
  &=& \sum_{i \in U} c_i + \sum_{ (i,j) \in E} q^{0}_{ij} - \sum_{ (i,j) \in E, i\in U} q^{0}_{ij} - \sum_{ (i,j) \in E, j\in U} q^{0}_{ij} + \sum_{ (i,j) \in E_2(U)} q^{0}_{ij}, \\
  &=&  \sum_{i \in U} (c_i- \sum_{(i,j) \in E}q_{ij}^0) +  \sum_{ (i,j) \in E_2(U)}q_{ij}^0 + \sum_{ (i,j) \in E} q^{0}_{ij},\\
   &=&  \phi(U) + \sum_{ (i,j) \in E} q^{0}_{ij}.
\end{eqnarray*}

 Thus an optimal solution to this UBQP is an optimal solution to GVC1.\\

  Let $\matr{Q}$ be the matrix defined as in equation (\ref{QGVC}) and $\vec{a}$ as in equation (\ref{aGVC}).

\begin{cor} \label{corGVC}
\begin{itemize}
\item[1.]  GVC is polynomial-time solvable on a  series-parallel graph .

\item[2.] If the rank of $\matr{Q}$ is fixed and $\matr{Q}$ is a positive semi-definite, then GVC is polynomial-time solvable.

\item[3.] If $\matr{Q}$ and $\vec{a}$ are non-negative, then GVC is polynomial-time solvable.

\item[4.] GVC is polynomial-time solvable if $q_{ij} \leq 0, ~ i \not= j$, where $\matr{Q}$ is defined in equation (\ref{QGVC}). 
\end{itemize}
\end{cor}
\begin{proof} From the construction given in equation (\ref{GVCUBQP1.1}), the support graph of $\matr{Q}$ is a subgraph of $G$ on which GVC is defined. Note that UBQP is solvable in polynomial time if the support graph of the cost matrix $\matr{Q}$ is series-parallel \cite{Francisco1996, 2010duan}. Since any subgraph of a series-parallel graph is series-parallel, the result follows.\\

If the rank of $\matr{Q}$ is fixed and $\matr{Q}$ is a positive semi-definite, then UBQP is polynomial time solvable \cite{2001Allemand, 2005Ferrez, 2010Kary,2010duan}, and if $\matr{Q}$ and $\vec{a}$ are non-negative, UBQP is polynomial time solvable \cite{2010duan}. The result now follows from Theorem \ref{gen3}.\\

When $q_{ij} \leq 0, ~ i \not= j$, where $\matr{Q}$ is defined in equation (\ref{QGVC}),  Picard et al.  \cite{1973Picard, 1975Picard} showed that UBQP  can be reduced to minimum-cut problem  on the graph $G=(V,E)$, where $V=(s,1,2,\hdots,n,t)$, $s$ denotes the source and $t$ the sink, and $E = E_s \cup E_Q \cup E_t$ with  $E_s= \{(s,j): j=1,2,\hdots,n \},~ E_Q= \{(i,j): q_{ij}  < 0, ~ 1 \leq i \leq j \leq n \},~E_t= \{ (j,t): j=1,2,\hdots,n\} $ \cite{2010duan}.  Therefore, in this case GVC is equivalent to the minimum-cut problem which is polynomial time solvable \cite{1973Picard, 1975Picard}. \\
\end{proof}

 Similar results follow  for GVC2 if $\matr{Q}$  is defined as in (\ref{QGVC2}) and $\vec{a}$ is defined as  in equation (\ref{aGVC2}) and for GVC1 when if $\matr{Q}$  is defined as in (\ref{QGVC1}) and $\vec{a}$ is defined as  in equation (\ref{aGVC1}).\\

It may be noted that the reductions discussed in Theorem \ref{gen2} and \ref{gen3} above for GVC and GVC1 to UBQP do not preserve $\epsilon$-optimality because of the resulting constant terms in the objective function. However, for GVC2, reduction to UBQP preserves  $\epsilon$-optimality. The equivalence however preserves other performance measures such as differential approximation ratio \cite{1980Ausiello, 1996Demange} and domination ratio \cite{1997Fred}.\\

For any graph $G=(V,E)$ and $v \in V$, $G-v$ is the graph obtained by deleting node $v$ and all its incident edges. Then GVC2 on $G$  can be solved by solving GVC2 on  $G-v$ with two different data sets.  Note that in an optimal solution to GVC2 on $G$, either $v$ belongs  or $v$ does not belongs to an optimal solution. If $v$ does not belong, then any optimal solution  to GVC2 on $G-v$ with original data (restricted to nodes and edges of $G-v$) is also an optimal solution to GVC2 on $G$. Otherwise, define

\begin{equation*}
c^{'}_i =
\begin{cases}
c_i + q_{iv}^2 ,  & \forall i \in \delta(v) ,\\
c_i  &\text{ otherwise }
\end{cases}
\end{equation*}

\noindent and $ q^{2^{'}}_{ij} = q^{2}_{ij}$ for all edge $(i,j)$ in $G-v$, where $\delta(v) = \{i : (i,v) \in  E \}$.

Let $U^{'}$ be an optimal solution to GVC2 on $G-v$ with $c_i$ replaced by $c^{'}_i$ and $q^{2}_{ij}$ is replaced by $q^{2^{'}}_{ij} $. Also let $U^0$ be an optimal solution to GVC2 on $G-v$ with original data. Then the best of $U^{'}, U^0$ will be an optimal solution to GVC2 on $G$. This idea  (together with Theorem \ref{gen2}) can be applied recursively to establish the following.

\begin{thm} GVC, GVC1, and GVC2 can be solved in polynomial time on $G$ if GVC2 can be solved in polynomial time  on $G- \{v_1,v_2,\hdots,v_k\}$ for some $\{v_1,v_2,\hdots,v_k\} \subseteq V$ where $k= O(\log n)$.\\
\end{thm}

It is well known that if $U \subseteq V$ is an optimal vertex cover then $V-U$ is an optimal independent set. A similar relationship exist between GVC1 and GVC2.

\begin{lem} $U^{0} \subseteq  V$ is an optimal solution to GVC1 with data $c_i, ~ \forall i \in V$ and $q^{0}_{ij},~ \forall (i,j) \in E$ if and only if $V - U^{0}$ is an optimal solution to GVC2 with data $-c_i, ~ \forall i \in V$ and $q^{2}_{ij}=q^{0}_{ij},~ \forall (i,j) \in E$
\end{lem}

\begin{proof}Suppose $U^{0} \subseteq  V$ is an optimal solution to the GVC1. Then for any $U \subseteq V$

\begin{align}
\sum_{i \in U^{0}} c_i + \sum_{(i,j) \in E_0(U^0)} q^{0}_{ij} & \leq \sum_{i \in U} c_i + \sum_{(i,j) \in E_0(U)} q^{0}_{ij}. \label{fes1}
\end{align}
Consider the objective function value of the GVC2 given in the statement of the lemma for the solution $V-U^0$. We have,

\begin{align*}
\sum_{i \in (V - U^{0})} -c_i + \sum_{(i,j) \in E_2(V-U^0)} q^{2}_{ij} & = \sum_{i \in V} -c_i - \sum_{i \in U^0} -c_i + \sum_{(i,j) \in E_0(U^0)} q^{2}_{ij},\\
&= \sum_{i \in V} -c_i + \sum_{i \in U^0} c_i + \sum_{(i,j) \in E_0(U^0)} q^{0}_{ij},\\
 & \leq \sum_{i \in V} -c_i + \sum_{i \in U} c_i + \sum_{(i,j) \in E_0(U)} q^{0}_{ij}, ~  \forall U \subseteq  V  \text{ from } (\ref{fes1}),\\
 & \leq \sum_{i \in (V-U)} -c_i  + \sum_{(i,j) \in E_2(V-U)} q^{2}_{ij}, ~  \forall U \subseteq  V. \\
\end{align*}

Since any subset of $V$ can be represented as $V-U$, for $ U \subseteq V$, it follows that  $V-U^0$ is an optimal solution to GVC2 defined above. The converse can be proved analogously.
\end{proof}

%

\subsection{Bipartite graphs}

Bipartite graph is a subclass of perfect graphs on which the MWVCP and MWISP can be solved in polynomial time \cite{2001Varizani}. Houchbaum  \cite{2002Dorit} and Bar-Yehuda et al. ~\cite{2010Reuven} showed that GVC1-HB  solvable in polynomial time on   bipartite graphs. Hochbaum and Pathria \cite{1997Dorit}  showed that GVC1-HP  is solvable in polynomial time on   bipartite graphs.  Thus it is interesting to investigate the complexity of GVC, GVC1, and GVC2 on bipartite graphs. Let us first define these problems on bipartite graphs using the bi-partitioned structure of the vertex set. \\

Let $G=(V_1,V_2,E)$ be a bipartite graph, where $V_1=\{1,2,\hdots,m\}$, $V_2=\{1,2,\hdots,n\}$ and $E=\{(i,j): i \in V_1, j \in V_2 \}$. For each vertex $i \in V_1$ a weight $c_i$ is prescribed and  for each vertex $j \in V_2$ a weight $d_j$ is prescribed. Let $ U_1  \subseteq V_1$,~ $ U_2  \subseteq  V2 $,  $E_0(U_1,U_2) = \{(i,j) \in E: i \not\in U_1,j \not\in U_2 \}$, $E_1(U_1,U_2) = \{(i,j) \in E: \text{ either } i \in U_1 \text{ or } j \in U_2 \text{ but not both }\} $, and  $E_2(U_1,U_2) = \{(i,j) \in E: i \in U_1,j \in U_2 \}$. Then GVC  on this bipartite graph $G$ is to find $U_1 \subseteq  V_1, U_2 \subseteq  V_2$ such that

$$  ~ f(U_1,U_2) = \sum_{i \in U_1}  c_i + \sum_{j \in U_2}  d_j   + \sum_{ (i,j) \in E_0(U_1,U_2)}q_{ij}^0 + \sum_{(i,j) \in E_1(U_1,U_2)}q_{ij}^1 + \sum_{(i,j) \in E_2(U_1,U_2)}  q_{ij}^2, $$

is minimized. Similarly,  GVC1  is to minimize
$$ ~ g(U_1,U_2) = \sum_{i \in U_1}  c_i + \sum_{j \in U_2}  d_j  +\sum_{ (i,j) \in E_0(U_1,U_2)}q_{ij}^0, $$ and GVC2 is to minimize
$$ ~ h(U_1,U_2) = \sum_{i \in U_1}  c_i + \sum_{j \in U_2}  d_j + \sum_{(i,j) \in E_2(U_1,U_2)} q_{ij}^2.  $$

Without loss of generality we assume $m \leq n$.

\begin{observation}GVC is solvable in polynomial time on bipartite graphs if   $q_{ij}^2-2q_{ij}^1+ q_{ij}^0 \geq 0,~  \forall (i,j)  \in E$ and  $~ c_i + d_j + \sum_{(i,j) \in E}(q_{ij}^2-q_{ij}^1) \geq 0,~ \forall ~ i \in V $ or $q_{ij}^2-2q_{ij}^1+ q_{ij}^0 \geq 0,~  \forall (i,j)  \in E$ and  $~ c_i+ d_j+ \sum_{(i,j) \in E}(q_{ij}^1-q_{ij}^0) \leq 0,~ \forall ~ i \in V $.
\end{observation}
\begin{proof}  If the first set of conditions are satisfied, then from Theorem \ref{gen2}, GVC reduces to GVC1-HB, which is solvable  in polynomial time on bipartite graphs \cite{2010Reuven}. If the second set of conditions  are satisfied, then GVC reduces to GVC2-HP, which is solvable in polynomial time on bipartite graphs \cite{1997Dorit}.
\end{proof}



\begin{thm} GVC, GVC1, and GVC2 on bipartite graphs are equivalent to BQP01. \label{gen4}
\end{thm}
\begin{proof}  For any feasible solution $U_1 \cup U_2$, where $U_1 \subseteq V_1$ and $U_2 \subseteq V_2$, let $h(U_1,U_2)$ and $\phi(U_1,U_2)$  be the objective function values of  GVC2, and BQP01 respectively.\\

Define the matrix $\matr{Q} = (q_{ij})_{m \times n}$ as

 \begin{equation}
q_{ij} =
\begin{cases}
q_{ij}^2 ,  & \forall (i,j) \in E,\\
0  &\text{ if }  (i,j) \not\in E.
\end{cases} \label{BQGVC2}
\end{equation}

\noindent   Set $a_i = c_i, ~ \forall i \in V_1$ and $b_j = d_j, ~ \forall j \in V_2$. Consider  the BQP01 with $\matr{Q}$, $\vec{a}$ and $\vec{b}$ are as defined above. For any feasible solution $ U_1 \cup U_2 $ of GVC2 (or BQP01), $h(U_1,U_2)=\phi(U_1,U_2)$. Thus an optimal solution to the BQP01 constructed above  is also an optimal solution to GVC2. \\

Let $\matr{Q}$ be the cost matrix of an instance  of BQP01 and $G(\matr{Q})$ be the support bipartite graph of $\matr{Q}$. Cost $a_i, ~i=1,2, \hdots,m$ is chosen as weight of vertex $i \in V_1$, cost $b_j, ~j=1,2, \hdots,n$ is chosen as weight of vertex $j$ for $j \in V_2$ and the quadratic cost $q_{ij},~ i=1,2,\hdots,m, j=1,2,\hdots,n$ is chosen as the cost of selecting both end points of edge $(i,j)$ of  $G(\matr{Q})$. Then, an optimal solution to the  resulting GVC2 on the bipartite graph $G(\matr{Q})$ solves  BQP01. The equivalence of BQP01 and GVC (GVC1) now follows from Theorem \ref{gen2}.
\end{proof}

As in the case of a general graph, to identify some interesting polynomially solvable special cases, we present below explicit re-formulations of GVC and GVC1 on the bipartite graph $G=(V_1 \cup V_2, E)$ as a BQP01. For any feasible solution $ U_1 \cup U_2$, let $f(U_1,U_2), g(U_1,U_2),$ and $\phi(U_1,U_2)$  be the objective function values of these  GVC, GVC1,  and BQP01 respectively.\\

Define $\matr{Q} = (q_{ij})_{m \times n}$ as

 \begin{equation}
q_{ij} =
\begin{cases}
q_{ij}^2 - 2q_{ij}^1 + q_{ij}^0,  & \forall (i,j) \in E,\\
0  &\text{ if }  (i,j) \not\in E.
\end{cases} \label{BQGVC}
\end{equation}
 $a_i = c_i + \sum_{(i,j) \in E} \left( q_{ij}^1 - q_{ij}^0 \right) ~ \forall i\in V_1$, and $b_j = d_j + \sum_{(i,j) \in E} \left( q_{ij}^1 - q_{ij}^0 \right) ~ \forall j\in V_2$.  Consider the BQP01 with $\matr{Q}$,  $\vec{a}$, and $\vec{b}$ are as defined above.  For any feasible solution $ U_1 \cup U_2 $\\

 \begin{equation}
\label{GVCBQP} f(U_1,U_2) =  \sum_{i \in U_1}  c_i + \sum_{j \in U_2}  d_j  + \sum_{ (i,j) \in E_0(U_1,U_2)}q_{ij}^0 + \sum_{(i,j) \in E_1(U_1,U_2)}q_{ij}^1 + \sum_{(i,j) \in E_2(U_1,U_2)}  q_{ij}^2.
\end{equation}

It is easy to show that
\begin{eqnarray}
\label{GVCBQP1}  \sum_{ (i,j) \in E_0(U_1,U_2)} q^{0}_{ij} &=& \sum_{ (i,j) \in E} q^{0}_{ij} - \sum_{ (i,j) \in E, i\in U_1} q^{0}_{ij} - \sum_{ (i,j) \in E, j\in U_2} q^{0}_{ij} + \sum_{ (i,j) \in E_2(U_1,U_2)} q^{0}_{ij}  \text{ and }\\
 \label{GVCBQP2}  \sum_{ (i,j) \in E_1(U_1,U_2)} q^{1}_{ij} & =&  \sum_{ (i,j) \in E, i\in U_1} q^{1}_{ij} + \sum_{ (i,j) \in E, j\in U_2} q^{1}_{ij} -  \sum_{ (i,j) \in E_2(U_1,U_2)} q^{1}_{ij}.
\end{eqnarray}

Now using equations (\ref{GVCBQP}), (\ref{GVCBQP1}) and (\ref{GVCBQP2}), we have

 \begin{align*}
f(U_1,U_2) =&   \sum_{i \in U_1} c_i + \sum_{j \in U_2} d_j + \sum_{ (i,j) \in E} q^{0}_{ij} - \sum_{ (i,j) \in E, i\in U_1} q^{0}_{ij} - \sum_{ (i,j) \in E, j\in U_2} q^{0}_{ij} + \sum_{ (i,j) \in E_2(U_1,U_2)} q^{0}_{ij}  \\
& +   \sum_{ (i,j) \in E, i\in U_1} q^{1}_{ij} + \sum_{ (i,j) \in E, j\in U_2} q^{1}_{ij} -  \sum_{ (i,j) \in E_2(U_1,U_2)} q^{1}_{ij} + \sum_{ (i,j) \in E_2(U_1,U_2)} q^{2}_{ij}, \\
=& \sum_{i \in U_1} (c_i + \sum_{(i,j) \in E} \left( q_{ij}^1 - q_{ij}^0 \right)) + \sum_{j \in U_2} (d_j + \sum_{(i,j) \in E} \left( q_{ij}^1 - q_{ij}^0 \right)) \\
& + \sum_{ (i,j) \in E_2(U_1,U_2)} (q_{ij}^2 - 2q_{ij}^1 + q_{ij}^0) + \sum_{(i,j) \in E} q_{ij}^0, \\
=& \sum_{i \in U_1} a_i + \sum_{j \in U_1} b_i+ \sum_{ (i,j) \in E_2(U_1,U_2)} q_{ij} +  \sum_{(i,j) \in E} q_{ij}^0, \\
   =&  \phi(U_1,U_2) + \sum_{(i,j) \in E} q_{ij}^0.
\end{align*}

 Thus an optimal solution to the BQP01 constructed above is also an optimal solution to GVC on   $G=(V_1 \cup V_2, E)$.\\

Now, define the matrix $\matr{Q} = (q_{ij})_{m \times n}$ as

\begin{equation}
q_{ij} =
\begin{cases}
q_{ij}^0 ,  & \forall (i,j) \in E,\\
0  &\text{ if }  (i,j) \not\in E.
\end{cases} \label{BQGVC1}
\end{equation}
 and $a_i = c_i- \sum_{(i,j) \in E}q_{ij}^0, ~ \forall i \in V_1$, $b_j = d_j- \sum_{(i,j) \in E}q_{ij}^0, ~ \forall j \in V_2$, Consider  the BQP01 with $\matr{Q}$, $\vec{a}$ and $\vec{b}$  as defined above. \\

 Then using equation (\ref{GVCBQP1}) we get
\begin{align*}
  g(U_1,U_2) =& \sum_{i \in U_1} c_i + \sum_{j \in U_2} d_i + \sum_{ (i,j) \in E_0(U_1,U_2)} q^{0}_{ij}, \\
  =& \sum_{i \in U_1} c_i + \sum_{j \in U_2} d_i +   \sum_{ (i,j) \in E} q^{0}_{ij} - \sum_{ (i,j) \in E, i\in U_1} q^{0}_{ij} - \sum_{ (i,j) \in E, j\in U_2} q^{0}_{ij} + \sum_{ (i,j) \in E_2(U_1,U_2)} q^{0}_{ij},\\
=& \sum_{i \in U_1} (c_i - \sum_{(i,j) \in E}   q_{ij}^0)  + \sum_{j \in U_2} (d_j - \sum_{(i,j) \in E} q_{ij}^0 ) + \sum_{ (i,j) \in E_2(U_1,U_2)}  q_{ij}^0 + \sum_{(i,j) \in E} q_{ij}^0, \\
=& \sum_{i \in U_1} a_i + \sum_{j \in U_1} b_i+ \sum_{ (i,j) \in E_2(U_1,U_2)} q_{ij} +  \sum_{(i,j) \in E} q_{ij}^0, \\
   =&  \phi(U_1,U_2) + \sum_{(i,j) \in E} q_{ij}^0.
\end{align*}

 Thus an optimal solution to this BQP01 is an optimal solution to GVC1 on   $G=(V_1 \cup V_2, E)$.\\

In view of Theorem \ref{gen4}, we have the following results for GVC  on bipartite graphs.

\begin{cor} \label{cor1GVCBQP} GVC on the bipartite graph $G=(V_1,V_2,E)$ is NP-hard if $m=O(\sqrt[k]{n})$, for any fixed $k$ but  solvable in polynomial time  if   $m=O(\log n)$.
\end{cor}
\begin{proof} From Theorem \ref{gen4}, GVC on the bipartite graph $G=(V_1,V_2,E)$ is equivalent to a BQP01 with size $ m \times n$ and the quadratic cost matrix  in BQP01 is as defined in equation (\ref{BQGVC}). Since BQP01 of size $m \times n$ is NP-hard if $m=O(\sqrt[k]{n})$ and solvable in polynomial time if  $m=O(\log n)$ \cite{{2015Abraham}}, the result follows.
\end{proof}

\begin{cor} \label{cor2GVCBQP}
If the matrix $\matr{Q} = (q_{ij})_{m \times n}$  is as defined by equation (\ref{BQGVC}) and has rank $O(\log n)$  or is a square matrix of bandwidth $d$, then GVC is solvable in polynomial time.
\end{cor}

\begin{proof}The proof follows from Theorem \ref{gen4} and corresponding properties of the equivalent BQP01 \cite{2014Piyashat,2015Abraham}.
\end{proof}

It may be noted we have a  result similar to Corollary \ref{cor2GVCBQP} on general graphs (see Corollary  \ref{corGVC}) with the additional restriction that $\matr{Q}$ must be positive-semidefinite. The bipartite structure  allows us to relax this restriction.\\

The results in corollary \ref{cor1GVCBQP} and \ref{cor2GVCBQP} have corresponding  versions of GVC1 and GVC2 as well. \\

\section{Integer programming formulations}

It may be noted that integer programming formulations  of GVC, GVC1, and GVC2 can  be obtained from the linearization of the equivalent UBQP \cite{1986Adams}. However,  direct formulations are more intuitive which is discussed below. \\

Consider the binary variables $x_i$ for  $i\in V$, $y_{ij}, z_{ij}$ for $(i,j) \in E$ where $x_i =1$ if and only if vertex $i$ is selected; $y_{ij} =1$ if and only if both vertices $i$ and $j$ are selected and $(i,j) \in E$; $z_{ij} =1$ if and only if none of the vertices $i$ and $j$ are selected and $(i,j) \in E$.
For any edge $(i,j) \in E$, either both $i$ and $j$ are  selected or none of $i$ and $j$ are selected or exactly one of $i$ or $j$ are selected and this later case can be handled by considering variable    $w_{ij} = 1-y_{ij}-z_{ij}$ for $(i,j) \in E$. We need not consider $w_{ij}$ explicitly and could simply use $1-y_{ij}-z_{ij}$ instead. We thus denote the objective function of GVC as $f(\vec{x},\vec{y},\vec{z})= \sum_{i\in V} c_ix_i + \sum_{(i,j) \in E} \left( q_{ij}^2 - q_{ij}^1\right) y_{ij}+ \sum_{(i,j) \in E} \left( q_{ij}^0 - q_{ij}^1\right)z_{ij} + \sum_{(i,j) \in E} q_{ij}^1$. It is possible to relax the integrality restriction on the variable $y_{ij}, \forall (i,j) \in E$. Thus we can formulate GVC as a mixed  integer linear program:

\begin{align}
\nonumber \text{ GVC-IP: ~~}  \text{ Minimized } &  ~ \sum_{i\in V} c_ix_i + \sum_{(i,j) \in E} \left( q_{ij}^2 - q_{ij}^1\right) y_{ij}+ \sum_{(i,j) \in E} \left( q_{ij}^0 - q_{ij}^1\right)z_{ij} \\
\nonumber \mbox{Subject to:~ } & x_i + x_j   \leq 1 + y_{ij}, ~ \forall  (i,j) \in E,\\
\nonumber & x_i    \geq y_{ij}, ~ \forall  (i,j) \in E,\\
 \nonumber &  x_j   \geq y_{ij}, ~ \forall  (i,j) \in E,\\
\label{GVC-1.4}   &   z_{ij} = 1-x_i-x_j+y_{ij} , ~\forall ~ (i,j) \in E,\\
\nonumber  &   y_{ij} \geq 0 , ~\forall ~ (i,j) \in E,\\
 \nonumber & x_{i}, \in \{0,1\} ~ \forall i \in V, z_{ij} \in \{0,1\} ~ \forall  (i,j) \in E.
\end{align}

Note that the constant  $\displaystyle \sum_{(i,j) \in E} q_{ij}^1$ is removed from the objective function.
Corresponding  Integer programming formulation of GVC1 and GVC2 are given below:

\begin{align*}
 \text{ GVC1-IP: ~~}  \text{Minimized } &  ~ \sum_{i \in V} c_i x_i + \sum_{(i,j) \in E}  q_{ij}^0 z_{ij} \\
 \mbox{Subject to:~ } & x_i + x_j   \geq 1 - z_{ij}, ~ \forall  (i,j) \in E,\\
 & z_{ij} \leq 1-x_i, ~ \forall  (i,j) \in E,\\
  &  z_{ij} \leq 1-x_j, ~ \forall  (i,j) \in E,\\
  & x_{i},  z_{ij} \in \{0,1\} ~ \forall i \in V,  (i,j) \in E.
\end{align*}

\begin{align*}
\nonumber \text{ GVC2-IP: ~~}  \text{Minimized} &  ~ \sum_{i \in V} c_i x_i + \sum_{(i,j) \in E}  q_{ij}^2 y_{ij} \\
 \mbox{Subject to:~ } & x_i + x_j   \leq 1 + y_{ij}, ~ \forall  (i,j) \in E,\\
 & y_{ij} \leq x_i, ~ \forall  (i,j) \in E,\\
  &  y_{ij} \leq x_j, ~ \forall  (i,j) \in E,\\
  & x_{i},  y_{ij} \in \{0,1\} ~ \forall i \in V,  (i,j) \in E.
\end{align*}

 \noindent For GVC1-IP when $c_i \geq 0, \forall i \in V$ and $q_{ij}^0 \geq 0, \forall (i,j) \in E$ then we can relax the constraints $z_{ij} \leq 1-x_i $ and $z_{ij} \leq 1-x_j, ~ \forall  (i,j) \in E$ and in this case the formulation GVC1-IP  reduces to that of GVC1-HB given in  \cite{2002Dorit}. Similarly in GVC2-IP, when $c_i \leq 0, \forall i \in V$ and $q_{ij}^2 \geq 0, \forall (i,j) \in E$, we can relax the constraints $y_{ij} \leq x_i, y_{ij} \leq x_j ~ \forall  (i,j) \in E$ and will get the formulation of GVC2-HP given in  \cite{1997Dorit}. \\

We denote  the linear programming relaxation  of GVC-IP, GVC1-IP, and GVC2-IP respectively by GVC-LP, GVC1-LP, and GVC2-LP.  \\

In GVC-IP, eliminating constraint (\ref{GVC-1.4}) by substituting $ y_{ij} = -1 + x_i + x_j + z_{ij} $  and we obtain an equivalent GVC1-IP. Similarly,  eliminating constraint (\ref{GVC-1.4}) by substituting  $z_{ij} = 1-x_i-x_j+y_{ij}$, we obtain an equivalent GVC2-IP. The LP relaxations of GVC-IP, equivalent GVC1-IP,  and equivalent GVC2-IP,  are denoted by GVC-LP(E), GVC1-LP(E), and GVC2-LP(E). (Note that the constant removed is taken into consideration when comparing the LP relaxation value).

 \begin{observation}  For the given instance of GVC-IP,  the LP relaxations GVC-LP(E), GVC1-LP(E) and  GVC-LP(E) have the same objective function value. \label{lem1}
 \end{observation}

 Nemhauser and Trotter \cite{1975Nemhauser} gave an integer programming formulation of  the Minimum weight vertex cover problem.  They also observed that extreme points of the LP relaxation of their integer programming formulation of MWVCP  are  half-integral.    Houchbaum \cite{2002Dorit} extended this result to  GVC1-HB when  vertex weight and edge weight are non-negative. Hassin and Levin  \cite{2003Rafael}  proved that the $\vec{x}$ component of   all extreme points of the LP relaxation of an integer programming  formulation of  GVC-HL  are half-integral ( i.e.  $x_i \in \{0,\frac{1}{2},1\}~ \forall i \in V$).   We  now prove a corresponding result for GVC-LP. The proof technique is similar to that is given in  \cite{2003Rafael}.

\begin{thm} For all  extreme points $(\vec{x},\vec{y},\vec{z})$ of  GVC-LP, $\vec{x}$ is half-integral.. \label{GVC-HL1}
\end{thm}
\begin{proof} Take any feasible solution  $(\vec{x},\vec{y},\vec{z})$ of GVC-LP for which $\vec{x}$ is not half-integral. If we can prove that any feasible solution, for which $\vec{x}$ is not half-integral, can not be an extreme point of GVC-LP, then we are done.  \\

Define
\begin{align}
\label{vc10} E^{-} &= \{~i~|~  0 < x_i <  \frac{1}{2}\}\\
\label{vc11} E^{+} &= \{~i~|  ~\frac{1}{2} < x_i <  1\}
\end{align}

Since $\vec{x}$ is not half-integral, therefore for some of the nodes   we will have $x_i \in \{0,\frac{1}{2},1\}$ but not for all. Thus $ E^{+} \cup E^{-} $ is not empty. Now we will construct two new solutions $ (\vec{x}^1,\vec{y}^1,\vec{z}^1)  \text{ and } (\vec{x}^2,\vec{y}^2,\vec{z}^2)$ using  $(\vec{x},\vec{y},\vec{z})$  as follows:  define $\epsilon > 0$ and \\

\begin{equation*}
x_i^1 =
\begin{cases}
x_i + \epsilon &\text{ if  } i \in E^{-}\\
x_i - \epsilon &\text{ if  } i \in  E^{+}\\
x_i & \text{ otherwise  }
\end{cases}
\end{equation*}

\begin{equation*}
x_i^2 =
\begin{cases}
x_i - \epsilon &\text{ if  } i \in E^{-}\\
x_i + \epsilon &\text{ if  } i \in  E^{+}\\
x_i & \text{ otherwise  }
\end{cases}
\end{equation*}

Define $y_{ij}^1= \min\{x_i^1,x_j^1\},~y_{ij}^2=\min\{x_i^2,x_j^2\}$ and $z_{ij}^1= 1-x_i^1-x_j^1+y_{ij}^1,  \text{ and } z_{ij}^2= 1-x_i^2-x_j^2+y_{ij}^2$. It is easy to verify that $ (\vec{x}^1,\vec{y}^1,\vec{z}^1)  \text{ and } (\vec{x}^2,\vec{y}^2,\vec{z}^2)$ are feasible solutions to GVC-LP.

 We can see that $(\vec{x},\vec{y},\vec{z}) = \frac{1}{2} \left[ (\vec{x}^1,\vec{y}^1,\vec{z}^1)+   (\vec{x}^2,\vec{y}^2,\vec{z}^2) \right]$. Since $ E^{+} \cup E^{-} $ is not empty, means for  $ \epsilon > 0$ small enough, $ (\vec{x}^1,\vec{y}^1,\vec{z}^1)  \text{ and } (\vec{x}^2,\vec{y}^2,\vec{z}^2)$  are distinct feasible solutions. Since $(\vec{x},\vec{y},\vec{z})$ is a convex combination of two feasible solutions $ (\vec{x}^1,\vec{y}^1,\vec{z}^1)  \text{ and } (\vec{x}^2,\vec{y}^2,\vec{z}^2)$, therefore it can not be an extreme point. This completes the proof.

\end{proof}

The above Theorem  also indicates that for all extreme points of GVC1-LP  and GVC2-LP,  $\vec{x}$ is  half-integral.\\

Consider the rounding algorithm for GVC given below:

\begin{algorithm}[H]
\caption{The rounding algorithm for GVC}
\label{Ralgorithm}
\begin{algorithmic}[1]
\\Solve GVC-LP, and denote by $ (x, y,z)$ an optimal extreme point solution.
\\ Set $x_i^*=1$ if and only if $x_i \geq \frac{1}{2}$.
\\ For edge $(i,j)$, set $ y_{ij}^*=\min\{x_i^*,x_j^*\}$, and $z_{ij}^* = 1-x^{*}_i - x^{*}_i + y_{ij}^*$.
\\ Output $ (x^*, y^*,z^*)$ .
\end{algorithmic}
\end{algorithm}

For  $q_{ij}^0 \geq q_{ij}^1 \geq   q_{ij}^2  \geq 0, ~ \forall (i,j) \in E$,  a rounding based 2-approximation algorithm has been proposed in \cite{2003Rafael}.  The rounding algorithm does not provide 2-approximation for GVC if $q_{ij}^0,q_{ij}^1,\text{ and } q_{ij}^2,~ \forall (i,j) \in E $ are arbitrary. For example, consider the  instance of GVC as follows: Take a cycle of three vertices $1,2,3$. Let $c_1=c_2=c_3=1$, $q_{12}^2=2, q_{23}^2=3, q_{13}^2=4 \delta , \delta > 0, \text{ and } q_{12}^1=q_{23}^1=q_{13}^1=0,~ q_{12}^0=q_{23}^0=q_{13}^0=\infty$. Note that the property $q_{ij}^0 \geq q_{ij}^1 \geq   q_{ij}^2  \geq 0 ~ \forall (i,j) \in E$ does not hold.\\

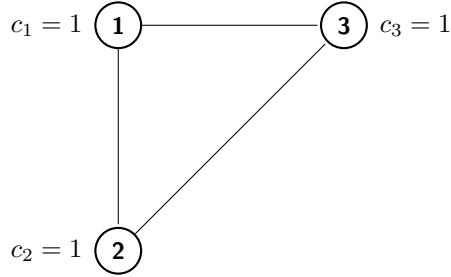
\begin{figure}[H]
\begin{center}
\begin{tikzpicture}[-,>=stealth',shorten >=1pt,auto,node distance=3cm,
  thick,main node/.style={circle,draw,font=\sffamily\small\bfseries},xscale=0.5, yscale=0.5]

  \node[main node] (1) {1};
  \node[main node] (2) [below   of=1] {2};
  \node[main node] (3) [ right of=1] {3};

\draw[thin] (1)node[left]{}-- node[midway]{} (2)node[left]{};
  \draw[thin] (1)node[left]{$c_1 = 1 \quad$}-- node[midway]{} (3)node[right]{$\quad c_3 =1$};
  \draw[thin] (2)node[left]{$c_2 = 1 \quad$}-- node[midway, right]{} (3)node[right]{};

\end{tikzpicture}
\end{center}
\end{figure}
An optimal solution to this GVC  is,  $x_1=1, x_2 =1, x_3=0, y_{12}= 1, y_{23}=0, y_{13}=0,z_{12}= 0, z_{23}=0, z_{13}=0$, with the objective value 4. Solution to the LP relaxation of the corresponding  GVC-IP is $x_1=\frac{1}{2}, x_2=\frac{1}{2}, x_3=\frac{1}{2},  y_{12}=0, y_{23}=0, y_{13}=0,z_{12}= 0, z_{23}=0, z_{13}=0$ with the objective value  1.5. Solution after applying the rounding strategy given in  Algorithm \ref{Ralgorithm} is $x_1=1, x_2 =1, x_3=1, y_{12}= 1, y_{23}=1, y_{13}=1,z_{12}= 0, z_{23}=0, z_{13}=0$ with the objective value $8 + 4 \delta$, which is $2+ \delta$-approximate. This example shows that rounding algorithm  can not guarantee 2-approximate solution if  the condition $q_{ij}^0 \geq q_{ij}^1 \geq   q_{ij}^2  \geq 0 ~ \forall (i,j) \in E$ does not hold .\\

We now extend the results of \cite{2003Rafael} to a more general class of problems. Let $(x^o,y^o,z^o)$ is an optimal solution for GVC.

\begin{lem} The Rounding algorithm for GVC produces a $\max\{2,\alpha, \alpha \beta \}$-approximate solution if $ q^2_{ij}  \leq \alpha q^1_{ij},$ $~ q^1_{ij}  \leq \beta q^0_{ij}$, $~ \alpha \geq 1, \beta \geq 1$ and $ q^0_{ij}, q^1_{ij},$ and $  q^2_{ij}$ are non-negative. \label{r2}
\end{lem}
\begin{proof}

Let $(\bar{x},\bar{y},\bar{z})$ is an optimal solution for GVC-LP and $(x^o,y^o,z^o)$ is an optimal solution for GVC-IP. Then,

\begin{align}
 \label{r2.0}  \sum_{i=1}^n c_i\bar{x}_i +  \sum_{(i,j) \in E } \left(q_{ij}^2 \bar{y}_{ij} +  q_{ij}^1 (1-\bar{z}_{ij}- \bar{y}_{ij})+ q_{ij}^0 \bar{z}_{ij} \right) & \leq \sum_{i=1}^n c_ix_i^o +  \sum_{(i,j) \in E } \left(q_{ij}^2 y_{ij}^o +  q_{ij}^1 (1-z_{ij}^o-y_{ij}^o)+ q_{ij}^0 z_{ij}^o \right).
\end{align}

Let $(x^{*},y^{*},z^{*})$ be the solution obtained by the rounding scheme given in Algorithm \ref{Ralgorithm}.  Since each  $\bar{x}_i$ changed by a factor of 2 at most, we get

\begin{equation}
\label{r2.1} \sum_{i=1}^n c_ix_i^*  \leq  2\sum_{i=1}^n c_i \bar{x}_i ~~~ \text{ by rounding}.
\end{equation}

Due to rounding $ \bar{x}_i \leq x^{*}_i, ~ \forall i \in V$. Since $\bar{y}_{ij}= \min \{\bar{x}_i, \bar{x}_j \} $ and $ y^{*}_{ij} =\min \{x^{*}_i,x^{*}_j \} $, therefore, $\bar{y}_{ij} \leq y^{*}_{ij}$. And $ z^{*}_{ij} = 1 - x^{*}_i  - x^{*}_j   + y^{*}_{ij}$, $~ \bar{z}_{ij} = 1 - \bar{x}_i   -  \bar{x}_j  +  \bar{y}_{ij} $, therefore, $\bar{z}_{ij} \geq z^{*}_{ij}$.\\

Now we can claim that
\begin{align}
\nonumber  q^2_{ij} y^{*}_{ij} + q^1_{ij} (1 -  y^{*}_{ij} - z^{*}_{ij} ) + q^0_{ij} z^{*}_{ij} &
  \leq  q^2_{ij} y^{*}_{ij} + \alpha q^1_{ij} (1 -  y^{*}_{ij} - z^{*}_{ij} ) +  \alpha \beta q^0_{ij} z^{*}_{ij}, \text{ since } \alpha, \beta \geq 1 \\
\nonumber & \leq  q^2_{ij} \bar{y}_{ij} + \alpha q^1_{ij} (1 -  \bar{y}_{ij} - \bar{z}_{ij} ) +  \alpha \beta q^0_{ij} \bar{z}_{ij}\\
\label{r2.3} & \leq \max \{ \alpha, \alpha \beta \} \left( q^2_{ij} \bar{y}_{ij} + q^1_{ij} (1 -  \bar{y}_{ij} - \bar{z}_{ij} ) + q^0_{ij} \bar{z}_{ij} \right).
\end{align}

The second inequality holds, since  $\bar{y}_{ij} \leq y^{*}_{ij}$ and $\bar{z}_{ij} \geq z^{*}_{ij}$,
 the above sum is the convex combination of $q^2_{ij}, \alpha q^1_{ij}, \text{ and } \alpha \beta q^0_{ij}$,  and   $q^2_{ij} \leq  \alpha q^1_{ij} \leq \alpha \beta q^0_{ij}$ ($ q^2_{ij}  \leq \alpha q^1_{ij} $ and $ q^1_{ij}  \leq \beta q^0_{ij}$).\\

From equation (\ref{r2.1}) and (\ref{r2.3})
\begin{align*}
 & \sum_{i=1}^n c_ix_i^*  +   \sum_{(i,j) \in E} \left( q^2_{ij} y^{*}_{ij} + q^1_{ij} (1 -  y^{*}_{ij} - z^{*}_{ij} ) + q^0_{ij} z^{*}_{ij} \right) \\
& \leq  2\sum_{i=1}^n c_i \bar{x}_i  + \max \{ \alpha, \alpha \beta \} \sum_{(i,j) \in E} \left( q^2_{ij} \bar{y}_{ij} + q^1_{ij} (1 -  \bar{y}_{ij} - \bar{z}_{ij} ) + q^0_{ij} \bar{z}_{ij} \right),\\
& \leq  \max \{2,  \alpha, \alpha \beta \} \left[ \sum_{i=1}^n c_i \bar{x}_i  +  \sum_{(i,j) \in E} \left( q^2_{ij} \bar{y}_{ij} + q^1_{ij} (1 -  \bar{y}_{ij} - \bar{z}_{ij} ) + q^0_{ij} \bar{z}_{ij} \right) \right],\\
& \leq    \max \{2,  \alpha, \alpha \beta \} \left[ \sum_{i=1}^n c_ix_i^o +  \sum_{(i,j) \in E } \left(q_{ij}^2 y_{ij}^o +  q_{ij}^1 (1-z_{ij}^o-y_{ij}^o)+ q_{ij}^0 z_{ij}^o \right) \right],~ \text{ from  (\ref{r2.0})}.
\end{align*}

 This proves the lemma.
\end{proof}

When $ \alpha =1, \beta=1$, we get the condition $q^0_{ij} \geq q^1_{ij} \geq   q^2_{ij} \geq 0$ which is the same condition   used in \cite{2003Rafael}, and in this case the approximation ratio 2. Thus, the Algorithm \ref{Ralgorithm} is a proper generalization of the rounding algorithm of \cite{2003Rafael}.\\

If we consider $c_i \geq 0, ~ \forall i \in V~$, $q_{ij}^0,q_{ij}^1,q_{ij}^2 \in [K,\alpha K],~ \forall (i,j) \in E$,  where   $  K \geq 0 \text{ and } \alpha > 1$, then it is easy to show that the rounding algorithm for GVC is $\max\{2,\alpha\}$-approximate for given $K \geq 0 \text{ and }  \alpha > 1$.\\

In GVC2, assign $c_i= \gamma, \forall i \in V$ and $q_{ij}^2 = \delta, \forall (i,j) \in E$, we call it uniform GVC2 (UGVC2). Then UGVC2 is to find  a subset $ U \subseteq V$ such that $ \displaystyle h(U) = \sum_{i \in U} \gamma + \sum_{(i,j) \in E_2(U)} \delta$ is minimized. \\

\begin{observation} For a given graph $G$, when $\gamma = - \delta$ and $ \delta > 0$,  UGVC2 is equivalent to solving ISP on $G$, and when  $ \gamma = - \delta$ and $ \delta < 0$, UGVC2 is equivalent to solving VCP on $G$.  \label{UVC2lem1}
\end{observation}

The above observation can be easily verified. For the first case if any optimal solution is not an independent set, we can  construct an alternate optimal solution which is an independent set. And in the second case if any optimal solution is not a vertex cover, we can construct an alternate optimal solution which is a vertex cover.\\


%

\begin{lem} If $\bar{G}$ is a subgraph of $G$ and $\sum_{i \in V(\bar{G})}  x_i \leq k$ is a valid inequality for MWISP then $\sum_{i \in V(\bar{G})}  x_i - \sum_{(i,j) \in E(\bar{G})}  y_{ij} \leq k$ is a valid inequality for GVC2, where $V(\bar{G})$ is the node set of $\bar{G}$, $E(\bar{G})$ is the edge set of $\bar{G}$, and $k$ is a constant.
\end{lem}

\begin{proof} Let $ (\vec{x},\vec{y})$ be a feasible solution of the given GVC2. If all $y_{ij} =0, ~ \forall (i,j) \in E(\bar{G})$ then it is easy to check that $ \vec{x}$ is the corresponding feasible solution for MWISP and $ (\vec{x},\vec{y})$ satisfies the valid inequality $\sum_{i \in V(\bar{G})}  x_i -\sum_{(i,j) \in E(\bar{G})} y_{ij} = \sum_{i \in V(\bar{G})}  x_i \leq k$.\\

Construct a GVC2 using left hand side of the valid inequality as an objective function as follows:

\begin{align*}
 \text{ Maximize ~} &  ~ \sum_{i \in V(\bar{G})}  x_i -\sum_{(i,j) \in E(\bar{G})} y_{ij}\\
 \text{ Subject to: }  & x_i + x_j  -y_{ij} \leq 1,~ \forall  (i,j) \in E(\bar{G}),\\
  & y_{ij} \leq x_i, y_{ij} \leq x_j ~ \forall  (i,j) \in E(\bar{G}),\\
   & (\vec{x},\vec{y}) \in \{0,1\}^n.
\end{align*}

If we can show that objective function value of the above integer program is bounded above by  $k$ then we are done.  In UGVC2,  select $\gamma =-1$ and  $\delta =1$, then the UGVC2 is equivalent to the above integer program.  From Observation \ref{UVC2lem1}, this case of UGVC2 is  equivalent to solving ISP on $\bar{G}$, proves that this special case of GVC2 on $\bar{G}$ is equivalent to solving ISP on $\bar{G}$, therefore,

\begin{equation*}
\sum_{i \in V(\bar{G})}  x_i -\sum_{(i,j) \in E(\bar{G})} y_{ij} = \sum_{i \in V(\bar{G})}  x_i \leq k
\end{equation*}

This shows that $\sum_{i \in V(\bar{G})}  x_i - \sum_{(i,j) \in E(\bar{G})}  y_{ij} \leq k$ is a valid inequality for GVC2.

\end{proof}

We use this result in the next section to prove properties for some related problems.

\section{Other special cases}

In GVC, if we set $q^{0}_{ij} = \infty$ for all $(i,j) \in E$, an optimal solution will be a vertex cover whenever $c_i$ is finite for $i \in V$. Motivated by this observation, we consider the \textit{vertex cover problem with node and edge weight} (VCPNEW) that seeks  a vertex cover  $P \subseteq  V $ of $G$ which minimizes the cost $\displaystyle \sum_{i \in P}  c_i  + \sum_{ (i,j) \in E_1(P)}q_{ij}^1 + \sum_{ (i,j) \in E_2(P)}  q_{ij}^2$, where $\displaystyle E_1(P) = \{(i,j): \text{ either } i \in P \text{ or } j \in P \text{ but not both} \} $,  $\displaystyle E_2(P) = \{(i,j): i,j \in P \}$. Unlike GVC, feasible solutions of VCPNEW are vertex covers. When $q^1_{ij}=0$ for all $(i,j) \in E$, we call the resulting VCPNEW as the \textit{vertex cover problem with over cover penalties} (VCOP) and when $q^{2}_{ij} = 0$ for all $(i,j) \in E$,  we call the resulting problem as the \textit{vertex cover problem with under cover penalties} (VCUP). \\

An interpretation of VCPNEW can be given as follows.  There are $n$ cities, and  city $i$ is located at the vertex  $i$ of the graph $G$, and two cities $i$ and $j$ are directly  connected if there is an edge $(i,j) \in E$. A  Multinational company Com1 is planning to open warehouses at different cities such that for every $(i,j) \in E$, there is at least one warehouse, either at $i$ or at $j$ or both. There is a setup cost $c_i$, if Com1 is opening a warehouse at city $i$. And if for any edge $(i,j) \in E$, there are two  warehouses at both ends of this edge then  over-covering cost $q^2_{ij}$ occurred, and if there is one  warehouse exactly  at one end of this edge then  under-covering cost $q^1_{ij}$ occurred. Com1 is planning to open warehouses such that sum of the setup cost, over-covering cost, and under-covering cost is minimized. This problem  can be formulated as the VCPNEW.\\

\begin{thm} VCPNEW, VCOP, and VCUP are equivalent. \label{vcpnewthm1}
\end{thm}
\begin{proof}  We first show that VCPNEW can be formulated as VCOP. Given an instance of VCPNEW, define

\begin{align*}
    c_i^{'} = & c_i,~ \forall ~ i \in V  \text{ and }\\
     q_{ij}^{2^{'}}  =& q_{ij}^2-q_{ij}^1,~ (i,j) \in E.
\end{align*}

\noindent  Now, consider the instance of VCOP where $c_i$ is replaced by  $c_i^{'}$ and  $q_{ij}^{2}$ is replaced by $q_{ij}^{2^{'}}$. Then, for any  vertex cover $P \subseteq V$, the objective function of this VCOP  can be written as

 \begin{equation*}
  g(P)= \sum_{ i\in P} c_i^{'}  + \sum_{(i,j) \in E_2(P)}q_{ij}^{2^{'}} =   \sum_{ i\in P} c_i  + \sum_{(i,j) \in E_2(P)} ( q_{ij}^2-q_{ij}^1).
 \end{equation*}

The objective function $f(P)$ of VCPNEW is given by

\begin{align*}
f(P) =& \sum_{i \in P}  c_i + \sum_{(i,j) \in E_2(P)}  q_{ij}^2 + \sum_{ (i,j) \in E_1(P)}q_{ij}^1,\\
     =& \sum_{i \in P}  c_i + \sum_{(i,j) \in E_2(P)}  q_{ij}^2 + \sum_{ (i,j) \in E}q_{ij}^1 -\sum_{ (i,j) \in E_2(P)}q_{ij}^1,\\
     =& \sum_{i \in P}  c_i + \sum_{(i,j) \in E_2(P)}  (q_{ij}^2  - q_{ij}^1 ) + \sum_{ (i,j) \in E}q_{ij}^1,\\
    = & g(P) + \sum_{(i,j) \in E} q_{ij}^1.
\end{align*}

\noindent Thus, for any vertex cover $P \subseteq V$ of VCPNEW and the instance of VCOP constructed above, $f(P) - g(P) $ is a constant. Therefore, an optimal solution to this VCPNEW  will also be  an optimal solution to VCOP.\\

Since VCOP is a special case of VCPNEW, any VCOP can be formulated as VCPNEW, establishing equivalence between VCPNEW and VCOP.\\

To establish the equivalence between VCPNEW and VCUP, define
\begin{align*}
    \bar{c}_i = & c_i,~ \forall ~ i \in V  \text{ and } \\
     \bar{q}_{ij}^{1}  =& q_{ij}^1- q_{ij}^2,~ (i,j)  \in E.
\end{align*}

\noindent Consider the instance of  VCUP where $c_i$ is replaced by  $\bar{c}_i$ and  $q_{ij}^{1}$ is replaced by $\bar{q}_{ij}^{1}$. As in the previous case, we can show that an optimal solution to this VCUP provides an optimal solution to VCPNEW and the converse follows the fact that VCUP is a special case of VCPNEW.\\
\end{proof}

\begin{thm} \label{gen6} VCPNEW, VCOP, and VCUP are equivalent to the minimum weight vertex cover problem (MWVCP).
\end{thm}
\begin{proof} We first show that  VCPNEW can be formulated as MWVCP.  For any vertex cover $P \subseteq \{1,2,\hdots,n\}$, let $ f(P)$ and $\phi(P)$  be the objective function values of  VCPNEW and MWVCP respectively. \\

Define

\begin{equation*}
w_i = c_i + \sum_{(i,j) \in E} (q_{ij}^2 - q_{ij}^1)
\end{equation*}
 and consider the MWVCP with node weight $w_i, ~ \forall i \in V$.\\

For any vertex cover $ P \subseteq  \{1,2\hdots,n\}$ of VCPEW

\begin{equation}
\label{VCPNEWVC}   f(P) =  \sum_{i \in P} c_i  + \sum_{ (i,j) \in E_1(P) } q^{1}_{ij} + \sum_{ (i,j) \in E_2(P)} q^{2}_{ij}. \\
\end{equation}

Since $P$ is a vertex cover, at least one end point of any edge will be in $P$, therefore, it can be verified that
\begin{eqnarray}
\label{VCPNEWVC1}  \sum_{ (i,j) \in E_2(P)} q^{2}_{ij} &=&  \sum_{ (i,j) \in E, i\in P} q^{2}_{ij} + \sum_{ (i,j) \in E, j\in U} q^{2}_{ij}  - \sum_{ (i,j) \in E} q^{2}_{ij}, \text{ and }\\
 \label{VCPNEWVC2}  \sum_{ (i,j) \in E_1(P) } q^{1}_{ij} & =&   \sum_{ (i,j) \in  E} q^{1}_{ij} - \sum_{ (i,j) \in  E_2(P)} q^{1}_{ij},\\
\nonumber &=& 2\sum_{ (i,j) \in E} q^{1}_{ij} - \sum_{ (i,j) \in E, i\in P} q^{1}_{ij} - \sum_{ (i,j) \in E, j\in P} q^{1}_{ij} ~~~( \text{ from  (\ref{VCPNEWVC1})}).
\end{eqnarray}

From (\ref{VCPNEWVC}), (\ref{VCPNEWVC1}) and (\ref{VCPNEWVC2}), we have

\begin{align}
\nonumber f(P) =&  \sum_{i \in P} c_i + 2\sum_{ (i,j) \in E} q^{1}_{ij} - \sum_{ (i,j) \in E, i\in P} q^{1}_{ij} - \sum_{ (i,j) \in E, j\in P} q^{1}_{ij} \\
\nonumber &+   \sum_{ (i,j) \in E, i\in P} q^{2}_{ij} + \sum_{ (i,j) \in E, j\in P} q^{2}_{ij}  - \sum_{ (i,j) \in E} q^{2}_{ij}, \\
 \nonumber = & \sum_{i \in P} \left(c_i + \sum_{(i,j) \in E} \left( q_{ij}^2 - q_{ij}^1 \right) \right)  + \sum_{(i,j) \in E} (2q_{ij}^1 - q_{ij}^2),  \\
  \label{VCPNEWVC1.1}  =&  \phi(P) + \sum_{(i,j) \in E} (2q_{ij}^1 - q_{ij}^2).
\end{align}

Thus an optimal solution to the MWVCP constructed above is also an optimal solution to VCPNEW.\\

 Let $\vec{w}=(w_1,w_2,\hdots,w_n)$ be an instance of MWVCP. Choose $c_i=w_i,~ \forall i \in V$, and $q^2_{ij} = q_{ij}^1 = 0,~ $ for all  $(i,j) \in E$. Then, an optimal solution to the VCPNEW on $\vec{c},  \vec{q^2}, \text{ and } \vec{q^1}$ solves MWVCP. Thus VCPNEW and MWVCP are equivalent. Now the result follows from the Theorem \ref{vcpnewthm1}.
\end{proof}


It may be noted that the equivalence established in Theorem \ref{gen6} preserves optimality, domination ratio \cite{1997Fred}, and differential approximation ratio \cite{1996Demange}. The reduction provided in the proof do not preserve $\epsilon-$optimality. \\

For a given instance   of VCPNEW with data $\vec{c}, \vec{q^2}, \vec{q^1}$, equivalent MWVCP (MWVCP-EQV) is defined with node weights  $c_i^{'} = c_i + \sum_{(i,j) \in E} (q_{ij}^2 - q_{ij}^1) $ for all $i\in V$. Let  $P^o$ be an optimal vertex cover for MWVCP-EQV and for a given vertex cover $P$,~ $\phi(P)$ is the objective function value of  MWVCP-EQV.  Define $\delta = \dfrac{\sum_{(i,j) \in E} (2q_{ij}^1 - q_{ij}^2) }{\phi(P^o)}$. Now we can state

\begin{lem} If $P^*$ is an $\epsilon$-approximate vertex cover for MWVCP-EQV, then $P^*$ is  $\dfrac{\epsilon+ \delta}{1+ \delta}$-approximate vertex cover for VCPNEW whenever  $\sum_{(i,j) \in E} (2q_{ij}^1 - q_{ij}^2) \geq 0$,~  $c_i^{'}  \geq 0$ for all $i\in V$, where $\delta$ is defined above. \label{r3}
\end{lem}
\begin{proof} For a vertex cover $P$ of  VCPNEW, let $f(P)$ is the objective function value of VCPNEW  and $\phi(P)$ is the objective function value of  MWVCP-EQV. From equation (\ref{VCPNEWVC1.1}) we know that

\begin{align}
\label{ratio1}  f(P) =&  \phi(P) + \sum_{(i,j) \in E} (2q_{ij}^1 - q_{ij}^2).
\end{align}

Since $P^*$ is an $\epsilon$-approximate vertex cover for MWVCP-EQV, therefore,

\begin{align}
\label{ratio2} & \dfrac{\phi(P^*)}{\phi(P^o)} \leq \epsilon.
\end{align}

When $c_i^{'} \geq 0$ for all $i\in V$, then $\phi(P^*), \phi(P^o) \geq 0$ and  in addition if $\sum_{(i,j) \in E} (2q_{ij}^1 - q_{ij}^2) \geq 0$ then $ \delta  \geq 0$. If three constants $H,O,K \geq 0$,~ $\dfrac{H}{O} \leq \epsilon$, and $K=\delta O$ for $\delta \geq 0$, then it is easy to check that $\dfrac{H+K}{O+K} \leq \dfrac{\epsilon+ \delta}{1+ \delta}$. Let $H= \phi(P^*), ~ O =\phi(P^o) , \text{ and } K= \sum_{(i,j) \in E} (2q_{ij}^1 - q_{ij}^2)=\delta \phi(P^o)$.  Therefore, for the given conditions in the lemma and  using equations (\ref{ratio1}) and (\ref{ratio2}) we get

\begin{align}
\nonumber  & \dfrac{H+K}{O+K} = \dfrac{\phi(P^*)+ \sum_{(i,j) \in E} (2q_{ij}^1 - q_{ij}^2)}{\phi(P^o)+\sum_{(i,j) \in E} (2q_{ij}^1 - q_{ij}^2)} \leq \dfrac{\epsilon+ \delta}{1+ \delta},~~~ \Rightarrow ~~  \dfrac{f(P^*))}{f(P^o)} \leq \dfrac{\epsilon+ \delta}{1+ \delta}
\end{align}

which implies that $P^*$ is an $\dfrac{\epsilon+ \delta}{1+ \delta}$-approximate vertex cover to VCPNEW (this is a tighter bound than $\epsilon$). This completes the proof.
\end{proof}

The well known rounding algorithm for MWVCP is 2-approximate \cite{1982Dorit}, therefore, in the view of Theorem \ref{r3} this rounding algorithm provides 2-approximate solution for VCPNEW if  $\sum_{(i,j) \in E} (2q_{ij}^1 - q_{ij}^2) \geq 0$ and $c_i + \sum_{(i,j) \in E} (q_{ij}^2 - q_{ij}^1) \geq 0$ for all $i\in V$.\\

When $q^0_{ij} \geq q^1_{ij}\geq q^2_{ij} \geq 0$ for all $(i,j)\in E$,  a 2-approximation rounding algorithm  for GVC has been given  in \cite{2003Rafael}. Since $q^0_{ij} = \infty$ in the case of VCPNEW, from \cite{2003Rafael}, if $q^0_{ij} \geq q^1_{ij}\geq q^2_{ij} \geq 0$, VCPNEW can be solved by a 2-approximation algorithm. The condition given in Lemma \ref{r3} is more relaxed then  $q^0_{ij} \geq q^1_{ij}\geq q^2_{ij} \geq 0$ for all $(i,j)\in E$, but still guarantees a 2-approximation solution.\\

Let us now consider integer programming formulations of VCPNEW, VCOP, and VCUP. Consider the binary variables,

For each  node $i \in V$  define
 \begin{equation*}
x_i =
\begin{cases}
1 &\text{if node $i$ is in the vertex cover,}\\
0 &\text{otherwise,}
\end{cases}
\end{equation*}
 and for every edge $(i,j)$ define
\begin{equation*}
y_{ij} =
\begin{cases}
1 &\text{if node $i$ and $j$ both are in the vertex cover,}\\
0 &\text{otherwise,}
\end{cases}
\end{equation*}
and
\begin{equation*}
r_{ij} =
\begin{cases}
1 &\text{if exactly one of node $i$ or $j$ is in the vertex cover,}\\
0 &\text{otherwise.}\\
\end{cases}
\end{equation*}

Then VCPNEW, VCOP, and VCUP can be formulated respectively as:

\begin{align}\nonumber  \text{ VCPNEW-IP: ~ Minimize}   ~  & \sum_{i=1}^{n}c_ix_i + \sum_{(i,j) \in E} q_{ij}^1 r_{ij} +\sum_{(i,j) \in E} q_{ij}^2 y_{ij}, \\
\label{vd} \text{Subjct to: ~~} & x_i + x_j -y_{ij} = 1, ~\forall (i,j) \in E,\\
\label{vd1} & y_{ij} + r_{ij} = 1, ~\forall (i,j) \in E,\\
\nonumber & x_i \in \{0,1 \}, ~\forall i \in V, \\
\nonumber & y_{ij},r_{ij} \geq 0,  ~ \forall (i,j) \in E.
\end{align}

\begin{align*} \mbox{ VCOP-IP: ~ Minimize}    ~  & \sum_{i=1}^{n}c_ix_i+\sum_{(i,j) \in E} q_{ij}^2 y_{ij} \\
 \text{Subjct to: ~~} & x_i + x_j -y_{ij} = 1, ~ \forall (i,j) \in E,\\
 & x_i \in \{0,1 \}, ~ \forall i \in V, \\
 & y_{ij} \geq 0,  ~ \forall (i,j) \in E.
\end{align*}

\begin{align*}\nonumber \mbox{ VCUP-IP: ~ Minimize}    ~  & \sum_{i=1}^{n}c_ix_i+\sum_{(i,j) \in E} q_{ij}^1 r_{ij} \\
 \text{Subjct to: ~~} & x_i + x_j + r_{ij} = 2, ~ \forall (i,j) \in E,\\
 & x_i \in \{0,1 \}, ~ \forall i \in V, \\
 & r_{ij} \geq 0,  ~ \forall (i,j) \in E.
\end{align*}

 The LP relaxation of VCPNEW-IP, VCOP-IP, and VCUP-IP are denoted by  VCPNEW-LP, VCOP-LP, and VCUP-LP respectively.\\

Hochbaum \cite{2002Dorit} presented  a class of integer programs called  IP2 and  showed that the extreme points to the LP relaxations of such problems are half-integral, and the relaxations can be solved using network flow algorithm in $O(mn)$ time. VCPNEW, VCOP, and VCUP are in IP2, therefore, all  extreme point  of VCPNEW-LP, VCOP-LP, and VCUP-LP are half-integral.\\

When $q^1_{ij} \geq 0$, $q^0_{ij}=q^2_{ij}=0$ for all $(i,j) \in E$ and $c_i=0, ~ \forall i \in V$, then VCUP is equivalent to minimum-cut problem \cite{1998Papadimitriou,1973Picard} and when $q^1_{ij} \leq 0$, $q^0_{ij}=q^2_{ij}=0$ for all $(i,j) \in E$ and $c_i=0, ~ \forall i \in V$, VCUP is equivalent to max-cut problem \cite{1991Boros,1975Picard}. Therefore, the minimum-cut problem  and the max-cut problem  are special cases of VCUP.\\

Moreover  $q^1_{ij} \geq 0$, $q^0_{ij}=q^2_{ij}=0$ for all $(i,j) \in E$ and $\vec{c}$ arbitrary, still  VCUP is equivalent to minimum-cut problem. Therefore, in this case also VCUP is polynomial time solvable.\\





In GVC, if we set $q^{2}_{ij} = \infty$ for all $(i,j) \in E$, an optimal solution will be an independent set whenever $c_i$ is finite for $i \in V$. Motivated by this observation, we consider the \textit{independent set problem with node and edge weight} (ISPNEW) that seeks  an independent set  $P \subseteq  V $ of $G$ which maximizes the cost $\displaystyle \sum_{i \in P}  c_i + \sum_{i,j \in P, (i,j) \in E}  q_{ij}^0 + \sum_{i \text{ or} j \in P, (i,j) \in E}q_{ij}^1$, where where $\displaystyle E_1(P) = \{(i,j): \text{ either } i \in P \text{ or } j \in P \text{ but not both} \} $,  $\displaystyle E_0(P) = \{(i,j): i,j \not\in P \}$. Unilike GVC, feasible solutions of ISPNEW are independent sets. When $q^1_{ij}=0$ for all $(i,j) \in E$, we call the resulting ISPNEW as the \textit{independent set problem with over  penalties} (ISOP) and when $q^{2}_{ij} = 0$ for all $(i,j) \in E$,  we call the resulting problem as the \textit{independent set problem with under  penalties} (ISUP). \\

\begin{thm} ISPNEW, ISOP, and ISUP are equivalent. \label{ispnewthm1}
\end{thm}
\begin{proof}  We first show that ISPNEW can be formulated as ISOP. Given an instance of ISPNEW, define

\begin{align*}
    c_i^{'} = & c_i,~ \forall ~ i \in V  \text{ and }\\
     q_{ij}^{0^{'}}  =& q_{ij}^0-q_{ij}^1,~ (i,j) \in E.
\end{align*}

\noindent  Now, consider the instance of ISOP where $c_i$ is replaced by  $c_i^{'}$ and  $q_{ij}^{0}$ is replaced by $q_{ij}^{0^{'}}$ . Then, for any  independent set $P \subseteq V$, the objective function of this ISOP  can be written as

 \begin{equation*}
  g(P)= \sum_{ i\in P} c_i^{'}  + \sum_{(i,j) \in E_0(P)}q_{ij}^{0^{'}} =   \sum_{ i\in P} c_i  + \sum_{(i,j) \in E_0(P)} ( q_{ij}^0-q_{ij}^1).
 \end{equation*}

The objective function $f(P)$ of ISPNEW is given by

\begin{align*}
f(P) =& \sum_{i \in P}  c_i + \sum_{(i,j) \in E_0(P)}  q_{ij}^2 + \sum_{ (i,j) \in E_1(P)}q_{ij}^1,\\
     =& \sum_{i \in P}  c_i + \sum_{(i,j) \in E_0(P)}  q_{ij}^0 + \sum_{ (i,j) \in E}q_{ij}^1 -\sum_{ (i,j) \in E_0(P)}q_{ij}^1,\\
     =& \sum_{i \in P}  c_i + \sum_{(i,j) \in E_0(P)}  (q_{ij}^0  - q_{ij}^1 ) + \sum_{ (i,j) \in E}q_{ij}^1,\\
    = & g(P) + \sum_{(i,j) \in E} q_{ij}^1.
\end{align*}

\noindent Thus, for any independent $P \subseteq V$ of ISPNEW and the instance of ISOP constructed above, $f(P) - g(P) $ is a constant. Therefore, an optimal solution to this ISPNEW  will also be  an optimal solution to ISOP.\\

Since ISOP is a special case of ISPNEW, any ISOP can be formulated as ISPNEW, establishing equivalence between ISPNEW and ISOP.\\

To establish the equivalence between ISPNEW and ISUP, define
\begin{align*}
    c_i^{''} = & c_i,~ \forall ~ i \in V  \text{ and } \\
     q_{ij}^{1^{''}}  =& q_{ij}^1- q_{ij}^0,~ (i,j)  \in E.
\end{align*}

\noindent Consider the instance of  ISUP where $c_i$ is replaced by  $c_i^{''}$ and  $q_{ij}^{1}$ is replaced by $q_{ij}^{1^{''}}$. As in the previous case, we can show that an optimal solution to this ISUP provides an optimal solution to ISPNEW and the converse follows the fact that ISUP is a special case of ISPNEW.\\
\end{proof}

\begin{thm} \label{gen7} ISPNEW, ISOP, and ISUP are equivalent to the maximum weight independent set problem (MWISP).
\end{thm}

\begin{proof} We first show that  ISPNEW can be formulated as MWISP.  For any independent set $P \subseteq \{1,2,\hdots,n\}$, let $ f(P)$, $ g(P)$, $ h(P)$ and $\phi(P)$  be the objective function values of  ISPNEW, ISOP, ISUP, and MWISP respectively. \\

Define

\begin{equation*}
w_i = c_i + \sum_{(i,j) \in E}  (q_{ij}^1 -q_{ij}^0)
\end{equation*}
 and consider the MWISP with node weight $w_i, ~ \forall i \in V$.\\

For any independent $ P \subseteq  \{1,2\hdots,n\}$ of ISPNEW

\begin{equation}
\label{ISPNEWVC}   f(P) =  \sum_{i \in P} c_i   + \sum_{ (i,j) \in E_0(P)} q^{0}_{ij} + \sum_{ (i,j) \in E_1(P) } q^{1}_{ij}. \\
\end{equation}

Since $P$ is an independent set, it can be verified that
\begin{eqnarray}
\label{ISPNEWVC1}  \sum_{ (i,j) \in E_0(P)} q^{0}_{ij} &=&  \sum_{ (i,j) \in E} q^{0}_{ij} - \sum_{ (i,j) \in E, i\in P} q^{0}_{ij} - \sum_{ (i,j) \in E, j\in U} q^{0}_{ij}   \text{ and }\\
 \label{ISPNEWVC2}  \sum_{ (i,j) \in E_1(P) } q^{1}_{ij} & =&    \sum_{ (i,j) \in E, i\in P} q^{1}_{ij} + \sum_{ (i,j) \in E, j\in P} q^{1}_{ij}.
\end{eqnarray}

From (\ref{ISPNEWVC}), (\ref{ISPNEWVC1}) and (\ref{ISPNEWVC2}), we have

\begin{align}
\nonumber f(P) =&  \sum_{i \in P} c_i +  \sum_{ (i,j) \in E} q^{0}_{ij} -   \sum_{ (i,j) \in E, i\in P} q^{0}_{ij} - \sum_{ (i,j) \in E, j\in P} q^{0}_{ij} +  \sum_{ (i,j) \in E, i\in P} q^{1}_{ij} + \sum_{ (i,j) \in E, j\in P} q^{1}_{ij} , \\
 \nonumber = & \sum_{i \in P} \left( c_i + \sum_{(i,j) \in E}  (q_{ij}^1 -q_{ij}^0)  \right)  + \sum_{(i,j) \in E}  q_{ij}^0  \\
  \label{ISPNEWVC1.1}  =&  \phi(P) + \sum_{(i,j) \in E}  q_{ij}^0 .
\end{align}

Thus an optimal solution to the MWISP constructed above is also an optimal solution to ISPNEW.\\

MWISP can be formulated as each of these problems. From Theorem \ref{ispnewthm1} ISPNEW, ISOP, and ISUP are equivalent to each other, it is enough to show that ISPNEW can be formulated as MWISP.  Let $\vec{w}=(w_1,w_2,\hdots,w_n)$ be an instance of MWISP. Choose $c_i=w_i,~ \forall i \in V$, and $q^0_{ij} = q_{ij}^1 = 0,~ $ for all  $(i,j) \in E$. Then, an optimal solution to the ISPNEW on $\vec{c}, \vec{q^0}, \text{ and } \vec{q^1}$ solves MWISP. Thus ISPNEW and MWISP are equivalent.
\end{proof}

It may be noted that the equivalence established in Theorem \ref{gen7} is only upto optimality. The reduction provided in the proof do not preserve $\epsilon-$optimality. \\

Let us consider integer programming formulations of ISPNEW, ISOP, and ISUP. Consider the binary variables,

  \begin{equation*}
x_i =
\begin{cases}
1 &\text{if node $i$ is in the independent set,}\\
0 &\text{otherwise.}
\end{cases}
\end{equation*}
 and for every edge $(i,j) \in E$
\begin{equation*}
z_{ij} =
\begin{cases}
1 &\text{if node $i$ and $j$ none of them in the independent set,}\\
0 &\text{otherwise.}
\end{cases}
\end{equation*}
and
\begin{equation*}
r_{ij} =
\begin{cases}
1 &\text{if node $i$ or $j$  is in the independent set,}\\
0 &\text{otherwise.}
\end{cases}
\end{equation*}

Then ISPNEW, ISOP, and ISUP can be formulated as:

\begin{align}\nonumber  \text{ ISPNEW-IP: ~ Maximize}    ~  & \sum_{i=1}^{n}c_ix_i+\sum_{(i,j) \in E} q_{ij}^0 z_{ij}+ \sum_{(i,j) \in E} q_{ij}^1 r_{ij}, \\
\label{is3} \text{ Subject to: ~~} & x_i + x_j + z_{ij} = 1, ~\forall (i,j) \in E,\\
\label{is4} & z_{ij} + r_{ij} = 1, ~\forall (i,j) \in E,\\
\nonumber & x_i \in \{0,1 \}, ~\forall i \in V, \\
\nonumber & z_{ij},r_{ij} \in \{0,1\},  ~ \forall (i,j) \in E.
\end{align}

\begin{align}\nonumber \mbox{ ISOP-IP: ~ Maximize}    ~   & \sum_{i=1}^{n}c_ix_i+\sum_{(i,j) \in E} q_{ij}^0 z_{ij} \\
\nonumber \text{ Subject to: ~~} & x_i + x_j +z_{ij} = 1, \mbox{ for all } (i,j) \in E,\\
\nonumber & x_i \in \{0,1 \}, ~\forall i \in V, \\
\nonumber & z_{ij} \in \{0,1\},  ~ \forall (i,j) \in E.
\end{align}

\begin{align}\nonumber \mbox{ ISUP-IP: ~ Maximize}    ~   & \sum_{i=1}^{n}c_ix_i+\sum_{(i,j) \in E} q_{ij}^1 r_{ij} \\
\nonumber \text{ Subject to: ~~} & x_i + x_j =r_{ij}, \mbox{ for all } (i,j) \in E,\\
\nonumber & x_i \in \{0,1 \}, ~\forall i \in V, \\
\nonumber & r_{ij} \in \{0,1\},  ~ \forall (i,j) \in E.
\end{align}

The LP relaxation of ISPNEW-IP, ISOP-IP, and ISUP-IP are denoted by  ISPNEW-LP, ISOP-LP, and ISUP-LP respectively. ISPNEW, ISOP, and ISUP are in IP2, therefore, all  extreme point  ISPNEW-IP, ISOP-IP, and ISUP-IP  are half-integral.\\

We can write  ISPNEW as the minimization problem with some modification as follows:

\begin{align}\nonumber \text{ ISPNEW: ~ } \min   ~  & \sum_{i=1}^{n}c_i(1-x_i)-\sum_{(i,j) \in E} q_{ij}^0 y_{ij}- \sum_{(i,j) \in E} q_{ij}^1 (1-y_{ij}) -\sum_{i=1}^{n}c_i \\
\nonumber \text{ st ~~} & (1-x_i) + (1-x_j) - y_{ij} = 1, ~\forall (i,j) \in E,\\
\nonumber & (1-x_i) \in \{0,1 \}, ~\forall i \in V, \\
\nonumber & y_{ij} \in \{0,1\},  ~ \forall (i,j) \in E.
\end{align}

This structure of ISPNEW shows that if $(\vec{x},\vec{y})$ is a feasible solution for ISPNEW then $(\vec{1-x},\vec{y})$ is a feasible solution for VCPNEW.\\

and if $(\vec{x},\vec{y})$  is an optimal solution of  ISPEW with objective $\sum_{i=1}^{n}c_ix_i+\sum_{(i,j) \in E} q_{ij}^0 y_{ij}+ \sum_{(i,j) \in E} q_{ij}^1 (1-y_{ij})$, then  $(\vec{1-x},\vec{y})$  is an optimal solution of the corresponding VCPEW with objective   $\sum_{i=1}^{n}c_ix_i-\sum_{(i,j) \in E} q_{ij}^0 y_{ij}- \sum_{(i,j) \in E} q_{ij}^1 (1-y_{ij})$.

\section{Conclusion}

We have discussed the GVC  which  generalizes GVC1, GVC2, MWVCP, and MWISP and also showed that GVC, GVC1, GVC2 are  equivalent to the UBQP. Some solvable cases are identified and approximation algorithms are suggested for some special cases. We also studied GVC on bipartite graphs and identified some polynomial solvable cases. We showed the equivalence of GVC on bipartite graphs and  the BQP01. We presented integer programming formulations of GVC and related problems. We discussed VCPNEW and ISPNEW which are special cases of GVC. We also showed that the VCPNEW is equivalent to MWVCP and the ISPNEW is equivalent to MWISP. Investigation of this problem not only helped  us to understand the nature of this problem but also helped us to identify many polynomial time solvable case using the equivalence with UBQP.

\end{document}